\documentclass{article}
\usepackage[utf8]{inputenc}

\usepackage[german,french,english]{babel}
\usepackage[T1]{fontenc}

\usepackage{fullpage}

\usepackage{amsmath}
\usepackage{amssymb}
\usepackage{amsthm}
\usepackage{stmaryrd}
\usepackage{relsize}

\usepackage{faktor}

\newtheoremstyle{itstyle}
  {1.5\topsep}{1.5\topsep}                                          % space above, space below
  {\itshape}{}{\bfseries}                               % body font, indent (<\parindent>), head font
  {.}{.5em}                                  % punctuation after head, space after head
  {\thmname{#1} \thmnumber{#2}\thmnote{ (#3)}}  % head spec
\newtheoremstyle{notitstyle}
  {1.5\topsep}{1.5\topsep}                                          % space above, space below
  {}{}{\bfseries}                               % body font, indent (<\parindent>), head font
  {.}{.5em}                                  % punctuation after head, space after head
  {\thmname{#1} \thmnumber{#2}\thmnote{ (#3)}}  % head spec

\theoremstyle{itstyle}%margin}%itstyle}
\newtheorem{thm}{Theorem}[section]
\newtheorem{lem}[thm]{Lemma}
\newtheorem{prop}[thm]{Proposition}
\newtheorem{cor}[thm]{Corollary}

\theoremstyle{notitstyle}
\newtheorem{defn}[thm]{Definition}

\newtheorem{rmk}[thm]{Remark}

\newtheorem{eg}[thm]{Example}

\usepackage{mathrsfs}

\usepackage{graphicx}

\usepackage{tikz}
\usetikzlibrary{decorations.pathmorphing,positioning,arrows,calc,chains,fit,shapes,automata,decorations.pathreplacing,calligraphy}

\tikzset{
modal/.style={
shorten >=1pt,
shorten <=1pt,
auto,
node distance=1.5cm,
semithick
},
world/.style={circle,draw,minimum size=0.5cm,fill=gray!15},
point/.style={circle,draw,inner sep=0.5mm,fill=black},
reflexive above/.style={->,loop,looseness=7,in=120,out=60},
reflexive below/.style={->,loop,looseness=7,in=240,out=300},
reflexive left/.style={->,loop,looseness=7,in=150,out=210},
reflexive right/.style={->,loop,looseness=7,in=30,out=330}
}

\usepackage{prftree}

\usepackage{soul} % for highlighting, underlining, ...
\usepackage{verbatim}

\usepackage{hyperref}
\hypersetup{
    colorlinks=true,
    linkcolor=blue,
    citecolor=magenta,      
    urlcolor=cyan,
}

\newif\ifhideproofs
\ifhideproofs
\usepackage{environ}
\NewEnviron{hide}{}

\fi

\let\epsilon\varepsilon
\let\phi\varphi
\newcommand{\col}{\mathop{:}}

\newcommand{\etal}{et~al.}

\usepackage{wrapfig}
\usepackage{enumitem}

\title{Kleene Theorems for Lasso Languages and $\omega$-Languages}
\author{Mike Cruchten \\ The University of Sheffield}
\date{}

\begin{document}

\maketitle

%\tableofcontents

%\newpage

\begin{abstract}
  Automata operating on pairs of words were introduced as an
  alternative way of capturing acceptance of
  regular $\omega$-languages. Families of DFAs and lasso automata operating on such pairs
  followed,
  giving rise to minimisation algorithms, a Myhill-Nerode theorem and
  language learning algorithms.
  Yet Kleene theorems for such a well-established class are still missing.
  We introduce rational lasso
  languages and expressions, show a Kleene theorem for lasso languages
  and explore the connection between rational lasso and
  $\omega$-expressions, which yields a Kleene theorem for
  $\omega$-languages with respect to saturated lasso automata. For one
  direction of the Kleene theorems, we also provide a Brzozowski
  construction for lasso automata from
  rational lasso expressions.
\end{abstract}

\section{Introduction}
Lassos occur naturally in the study of $\omega$-automata, where they manifest
themselves in nondeterministic B\"{u}chi automata as paths consisting of a
prefix leading from the initial state to some state $q$, and a period which
leads from $q$ back to itself, traversing some accepting state infinitely often.
The existence of such paths for non-empty
non-deterministic B\"{u}chi automata is necessary, in that every such automaton
where some accepting state is reachable must admit an accepting run in the
shape of a lasso. On the level of words, these infinite paths correspond to
infinite words of the shape $uv^\omega$, which are called ultimately periodic words.
These words play an important role, as regular $\omega$-languages are
completely characterised by them in the sense that two regular $\omega$-languages
are equal if and only if they contain the same ultimately periodic words.
By representing an ultimately periodic word $uv^\omega$ as a string
$u\$v$, Calbrix, Nivat and Podelski show that for any regular $\omega$-language $L$,
the set $L_{\$}=\{u\$ v\mid uv^\omega \in L\} $ is regular and can be accepted by
a finite deterministic automaton (DFA) using an alphabet extended by the
symbol $\$$ \cite{calbrix:1994:ultimatelyPeriodicWords}. Additionally,
Calbrix \etal\ give conditions under which a regular language $U$ over
the extended alphabet is equal to $L_{\$}$ for some regular $\omega$-language
$L$. Their work shows that certain DFAs over an extended alphabet, called
$L_{\$}$-automata, act as acceptors of regular $\omega$-languages.
These $L_{\$}$-automata are deterministic and their acceptance condition makes no
longer use of infinite paths. The results came with
the hope of improving existing algorithms for deciding emptiness and language
inclusion of regular $\omega$-languages, which are prominently used in software
verification and model checking.

Angluin and Fisman use $(u,v)$ (instead of $u\$v$) as a representation of $uv^\omega$, and introduce families of
DFAs (FDFAs) operating on such pairs \cite{angluin:2016:learningRegularOmegaLanguages}. They combine the work by Calbrix
\etal\ on $L_{\$}$-automata and also work done by Maler and Staiger
\cite{maler:1997:onSyntacticCongruences} on syntactic congruences to
produce different variations of FDFAs, the periodic and syntactic FDFA.
They additionally define a recurrent FDFA and investigate the differences
in size for the various FDFAs. From these automata, they devise language
learning algorithms to learn regular $\omega$-languages and show that this can
be done in polynomial time in the size of the FDFA.
Moreover, in \cite{angluin:2016:FDFAs}, they investigate the
complexity of certain operations and decision procedures on FDFAs,
including deciding emptiness and language inclusion,
and the performance of Boolean operations. They show that these can all be
performed in non-deterministic logarithmic space, validating the hopes of Calbrix
\etal

An equivalent automaton to the FDFA is the lasso automaton
defined by Ciancia and Venema \cite{ciancia:2019:omegaAutomata}.
Lasso automata also operate on pairs $(u,v)$, which they call
lassos. We follow this convention and call pairs $(u,v)$ representing
ultimately periodic words lassos.
In \cite{ciancia:2019:omegaAutomata}, Ciancia \etal\ give a Myhill-Nerode
theorem and show that lasso automata can be
minimised using partition refinement.  Alternatively, minimisation can
also be obtained from a double reverse powerset construction \`{a} la
Brzozowski \cite{cruchten:2022:omegaAutomata}.

Although automata operating on lassos, such as $L_{\$}$-automata, FDFAs and
lasso automata, are well-established in many
regards, they still lack a Kleene theorem.
Our main goal is to establish a Kleene theorem for lasso languages, that is
sets of lassos, with respect to lasso automata, and to show how rational lasso and
$\omega$-expressions relate. This paves a way towards a Kleene theorem
for $\omega$-languages with respect to saturated lasso automata
(Definition \ref{defn:omegaAutomaton}).

%\begin{wrapfigure}{r}{0.55\textwidth}
%  \vspace{-1cm}
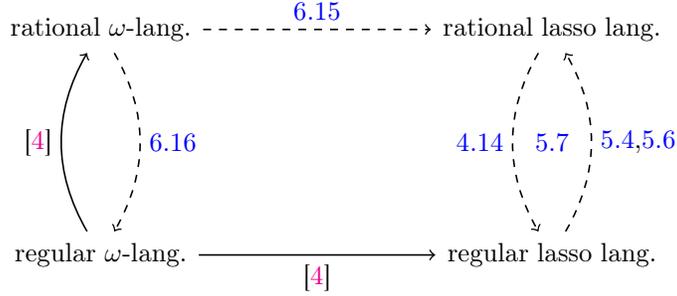
\begin{figure}[]
  \begin{center}
    \begin{tikzpicture}[modal]
        \node[] at (6,3) (a) [label=above:{}] {rational lasso lang.};
        \node[] at (6,0) (b) [label=above:{}] {regular lasso lang.};
        \node[] at (0,3) (c) [label=above:{}] {rational $\omega$-lang.};
        \node[] at (0,0) (d) [label=above:{}] {regular $\omega$-lang.};
        \node at (6,1.5) (e) {\ref{thm:KleeneTheoremLassoLanguages}};
        \path[->,dashed] (a) edge[bend right] node[left]{\ref{thm:BrzozowskiOmegaAutomaton}} (b);
        \path[->,dashed] (b) edge[bend right] node[right]{\ref{prop:regLassoFromLassoAutomaton},\ref{cor:ratLassoFromLassoAutomaton}} (a);
        \path[->] (d) edge node[below]{\cite{calbrix:1994:ultimatelyPeriodicWords}} (b);
        \path[->,dashed] (c) edge node[]{\ref{prop:rationalLassoFromOmega}} (a);
        \path[->,dashed] (c) edge[bend left] node[]{\ref{thm:omegaAutomatonFromOmegaExpression}} (d);
        \path[->] (d) edge[bend left] node[]{\cite{calbrix:1994:ultimatelyPeriodicWords}} (c);
      %\draw (current bounding box.south east) rectangle (current bounding box.north west);
    \end{tikzpicture}
  \end{center}
  \caption{Diagram showing our main contributions as dashed arrows.}\label{fig:contributions}
\end{figure}
%  \vspace{-1cm}
%\end{wrapfigure}

Our contributions are indicated as dashed arrows in Figure
\ref{fig:contributions}.  We define rational lasso languages as those
lasso languages that can be obtained from rational languages using
rational lasso operations.  Our first contribution is a Kleene theorem
for lasso languages: we show that a lasso language is rational if and
only if it is accepted by a finite lasso automaton (Theorem
\ref{thm:KleeneTheoremLassoLanguages}).  For one direction, we provide
a Brzozowski construction, which turns a rational lasso expression
into a finite lasso automaton accepting the corresponding rational
lasso language (Theorem \ref{thm:BrzozowskiOmegaAutomaton}).  For the
converse direction, we show how to
dissect a finite lasso automaton into several DFAs and prove that the
lasso language accepted by the lasso automaton can be obtained from
the rational languages corresponding to the DFAs by using the rational
lasso operations (Proposition \ref{prop:regLassoFromLassoAutomaton}
and Corollary \ref{cor:ratLassoFromLassoAutomaton}), following ideas
from \cite{calbrix:1994:ultimatelyPeriodicWords}.

Secondly, we study the relationship between rational lasso and
$\omega$-expressions.  We introduce a novel notion of rational lasso
expression representing a rational
$\omega$-expression (Definition \ref{defn:covering}). Intuitively, it
expresses that the language semantics of either expression completely
determines that of the other. We show that, for any given rational
$\omega$-expression, we can construct a representing rational lasso
expression in a syntactic manner provided we have access to two
additional operations on rational expressions (cf.\thinspace
Proposition \ref{prop:rationalLassoFromOmega}).

These two contributions, together with a result of
\cite{calbrix:1994:ultimatelyPeriodicWords}, allow us to re-establish
Kleene's theorem for $\omega$-languages with respect to saturated
lasso automata. Given a rational $\omega$-language, we can turn it
into a rational lasso expression (Proposition
\ref{prop:rationalLassoFromOmega}) and apply our Brzozowski
construction (Theorem \ref{thm:BrzozowskiOmegaAutomaton}) to obtain
the desired finite saturated lasso automaton, hence every rational
$\omega$-language is regular (Theorem
\ref{thm:omegaAutomatonFromOmegaExpression}).  The converse direction is
given by \cite{calbrix:1994:ultimatelyPeriodicWords}, showing that
every $\omega$-language accepted by a finite saturated lasso automaton
is rational.

\paragraph{Organisation of the paper.} Section \ref{sec:2} consists of the preliminaries.
In Section \ref{sec:3} we introduce rational lasso expressions and languages
together with an algebra. Section \ref{sec:4} is devoted to showing that every
rational lasso language is regular, for which we give a Brzozowski construction
for lasso automata. In Section \ref{sec:5} we show that every regular lasso language
is rational. Section \ref{sec:6} investigates the connection between rational lasso and
$\omega$-expressions.

\section{Preliminaries}\label{sec:2}

Throughout this article, $\Sigma$ denotes a finite alphabet. The free
monoid $(\Sigma^\ast,\cdot,\epsilon)$ over $\Sigma$
is formed by finite words, written $u,v,w$, with
concatenation $u\cdot v=uv$ and the empty word $\epsilon$. We write
$\Sigma^\omega$ for the set of infinite words and
$\Sigma^{\text{up}}$ for the set of ultimately periodic words,
that is, those of the form $uv^\omega$ with $v\not =\epsilon$.
A \emph{lasso} is a pair
$(u,v)\in\Sigma^{\ast}\times\Sigma^+$, with $u$ the \emph{spoke} and $v$ the
\emph{loop}. We write $\Sigma^{\ast+}$ for the set $\Sigma^\ast\times\Sigma^+$
and think of the lasso $(u,v)$ as a representative for $uv^\omega$.

We write $U,V,W$ for languages of words and $L,K$ for languages of
infinite words ($\omega$-languages) or of lassos, depending on
context.  As is standard, the rational languages are those obtained from
$\emptyset$, $\{a\}\ (a\in\Sigma)$ and $\{\epsilon\}$ using the rational operations
of language concatenation, union and Kleene star. Similarly, the rational
$\omega$-languages are those, which can be written as a finite union
$\bigcup_{i\in I}U_i\cdot V_i^\omega$ where the $U_i$ and $V_i$ are rational
languages, $\cdot$ denotes the concatenation between a language and an
$\omega$-language, and $(-)^\omega$ denotes $\omega$-iteration.
A language is regular if it is accepted
by a DFA and an $\omega$-language is regular if it is accepted by a
finite nondeterministic B\"{u}chi automaton.
For an $\omega$-language $L$, $\text{UP}(L)=L\cap \Sigma^{\text{up}}$ is its
\emph{ultimately periodic fragment}. The set $\text{Exp}$ of rational
expressions is given by the grammar:
\begin{align*}
  t &::= 0 \mid 1\mid a\in\Sigma\mid t\cdot t\mid t+t\mid t^\ast.
\end{align*}
We reserve the symbols $t,r,s$ for rational expressions. To each rational
expression we associate a rational language through the semantics map
$\llbracket -\rrbracket:\text{Exp}\to 2^{\Sigma^\ast}$, which is the unique
map satisfying:
\begin{align*}
  \llbracket 0\rrbracket &= \emptyset & \llbracket 1\rrbracket &=\{\epsilon\} & \llbracket a\rrbracket &= \{a\} \\
  \llbracket t\cdot r\rrbracket &= \llbracket t\rrbracket\cdot \llbracket r\rrbracket & \llbracket t+r\rrbracket &= \llbracket t\rrbracket \cup\llbracket r\rrbracket & \llbracket t^\ast\rrbracket &= \llbracket t\rrbracket^\ast
\end{align*}
We write $N$ for the set of all rational expressions which have the
\emph{empty word property}, that is those $t\in \text{Exp}$ such that $\epsilon\in\llbracket t\rrbracket$.
The sets $\text{Exp}_\omega$ of $\omega$-rational expressions are given by the grammar:
\begin{align*}
  T &::= 0_\omega\mid T+_\omega T\mid t\cdot_\omega T\mid r^\omega,
\end{align*}
where $t,r\in\text{Exp}$ and $r\not\in N$. We drop the subscript $(-)_\omega$ 
whenever this does not lead to confusion, thus simply writing $0, +$ and $\cdot$.
The symbol $T$ ($T_1,T_2,\ldots$) is reserved for rational $\omega$-expressions.
Similarly to the rational expressions, we associate to each rational $\omega$-expression
a rational $\omega$-language through the semantics map
$\llbracket-\rrbracket_\omega :\text{Exp}_\omega \to 2^{\Sigma^\omega}$ 
uniquely determined by
\begin{align*}
  \llbracket 0_\omega\rrbracket_\omega &= \emptyset & \llbracket T_1+_\omega T_2\rrbracket_\omega &= \llbracket T_1\rrbracket_\omega \cup \llbracket T_2\rrbracket_\omega & \llbracket t\cdot_\omega T\rrbracket_\omega &= \llbracket t\rrbracket \cdot \llbracket T\rrbracket_\omega & \llbracket r^\omega \rrbracket_\omega &= \llbracket r\rrbracket^\omega.
\end{align*}

The following definition introduces a rewrite system on lassos. The
rewrite system alongside some of its properties (which are stated below) can
be found in \cite{cruchten:2022:omegaAutomata}.

\begin{defn}
				We define the following rewrite rules on $\Sigma^{\ast +}$:
				\[
								\prftree[r]{($\gamma_1$)}{a\in\Sigma}{(ua,va)}{(u,av)}\quad\quad\quad\quad \prftree[r]{($\gamma_2$)}{(u,v^k)}{(k>1)}{(u,v)}
				\] 
				We write $(u,v)\to_{\gamma_i}(u',v')$ if $(u,v)$ rewrites
				to $(u',v')$ in one step under $\gamma_i$ and say that
        $(u,v)$ \emph{$\gamma_i$-reduces} to $(u',v')$ (or that $(u',v')$
        \emph{$\gamma_i$-expands} to $(u,v)$).
				We let ${\to_\gamma}={\to_{\gamma_1}}\cup{\to_{\gamma_2}}$, denote
        by $\sim_\gamma$ the least equivalence relation including $\to_{\gamma}$ and say
        that two lassos are \emph{$\gamma$-equivalent} if they are $\sim_{\gamma}$-related.
\end{defn}

\begin{prop}
  The relation $\to_{\gamma}$ is confluent and strongly normalising.
\end{prop}

It follows from this proposition that every lasso $(u,v)$ has a
unique \emph{normal form}. As an example, take the lasso
$(aba,baba)$, which is a representative for  $(ab)^\omega$. We show the different
reductions which lead to the normal form $(\epsilon,ab)$ below.
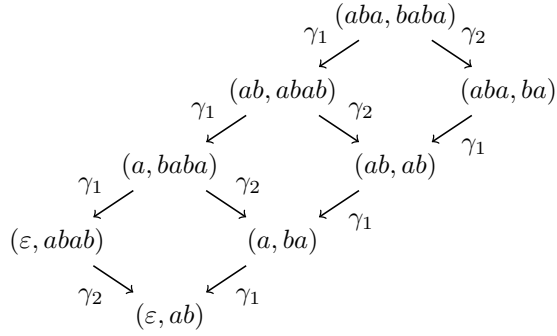
\begin{center}
\begin{tikzpicture}[modal]
\node[] at (1.5,0) (a) [label=above:{}] {$(\epsilon,ab)$};
\node[] at (0,1) (b) [label=above:{}] {$(\epsilon,abab)$};
\node[] at (3,1) (c) [label=above:{}] {$(a,ba)$};
\node[] at (1.5,2) (d) [label=above:{}] {$(a,baba)$};
\node[] at (4.5,2) (e) [label=above:{}] {$(ab,ab)$};
\node[] at (3,3) (f) [label=above:{}] {$(ab,abab)$};
\node[] at (6,3) (g) [label=above:{}] {$(aba,ba)$};
\node[] at (4.5,4) (h) [label=above:{}] {$(aba,baba)$};
\path[->] (b) edge node[below left]{$\gamma_2$} (a);
\path[->] (c) edge node[]{$\gamma_1$} (a);
\path[->] (d) edge node[]{$\gamma_2$} (c);
\path[->] (d) edge node[above left]{$\gamma_1$} (b);
\path[->] (e) edge node[]{$\gamma_1$} (c);
\path[->] (f) edge node[]{$\gamma_2$} (e);
\path[->] (f) edge node[above left]{$\gamma_1$} (d);
\path[->] (g) edge node[]{$\gamma_1$} (e);
\path[->] (h) edge node[]{$\gamma_2$} (g);
\path[->] (h) edge node[above left]{$\gamma_1$} (f);
\end{tikzpicture}
\end{center}
This picture suggests that if $(u,v)\to_\gamma (u',v')$, then
$uv^\omega = u'v'^\omega$, that is, reduction preserves the ultimately
periodic word which is represented. In fact, the other direction
also holds as shown by the next proposition.

\begin{prop}[Lasso Representation Lemma]
	Let $(u,v),(u',v')\in\Sigma^{\ast +}$. Then
  \[
    (u,v)\sim_\gamma (u',v') \iff uv^\omega = u'v'^\omega.
  \] 
\end{prop}

One consequence of this proposition, is that two lassos reduce to the
same normal form if and only if they are a representative of the same
ultimately periodic word.

\begin{defn}[\cite{ciancia:2019:omegaAutomata}]
  A \emph{lasso automaton} is a structure
  $\mathcal{A}=(X,Y,\overline{x},\delta_1,\delta_2,\delta_3,F)$ where
  $\delta_1\colon X\to X^\Sigma$, $\delta_2\colon X\to Y^\Sigma$,
  $\delta_3\colon Y\to Y^\Sigma$ and $F\subseteq Y$.
  We call $X$ and $Y$ the
  sets of \emph{spoke} and \emph{loop states}.  The maps
  $\delta_1,\delta_2$ and $\delta_3$ are called the \emph{spoke},
  \emph{switch} and \emph{loop transition maps} of $\mathcal{A}$.  The set $F$ denotes the
  \emph{accepting states} and $\overline{x}\in X$ is the \emph{initial
  state} of $\mathcal{A}$.
\end{defn}

We do not consider nondeterministic lasso automata in this text.
For convenience, we assume that $X$ and $Y$ are disjoint and define the map
$(\delta_2\col\delta_3):X\uplus Y\to Y^\Sigma$ which is equal to
$\delta_2$ on $X$ and equal to $\delta_3$ on $Y$.
  The maps
  $\delta_1,(\delta_2\col\delta_3)$ and $\delta_3$ can be extended
  from symbols to finite words in the usual way.

\begin{defn}[\cite{ciancia:2019:omegaAutomata}]
  A lasso $(u,v)\in\Sigma^{\ast +}$ is \emph{accepted} by
  $\mathcal{A}$, if
  $(\delta_2\col\delta_3)(\delta_1(\overline{x},u),v)\in F$.
  The set of
  all lassos accepted by
  $\mathcal{A}$ is denoted $L_\circ(\mathcal{A})$. A lasso language
  $L$ is \emph{regular} if it is accepted by a finite lasso automaton.
\end{defn}

\begin{figure}[htpb]
  \centering
  \begin{tikzpicture}[modal]
    \node[initial,state] at (0,3) (a) [label=above:{}] {$0$};
    \node[state] at (0,1) (b) [label=above:{}] {$1$};
    \node[state,accepting] at (2,4) (c) [label=above:{}] {$2$};
    \node[state] at (2,0) (d) [label=above:{}] {$3$};
    \node[state] at (4,2) (e) [label=above:{}] {$4$};
    \path[->] (a) edge node[]{$b$} (b);
    \path[->] (a) edge[reflexive above] node[above]{$a$} (a);
    \path[->,dotted] (a) edge node[]{$b$} (c);
    \path[->,dotted] (a) edge node[]{$a$} (e);
    \path[->] (b) edge[reflexive left] node[]{$a,b$} (b);
    \path[->,dotted] (b) edge node[below right]{$b$} (e);
    \path[->,dotted] (b) edge node[below left]{$a$} (d);
    \path[->,dashed] (c) edge node[]{$b$} (e);
    \path[->,dashed] (c) edge[reflexive right] node[right]{$a$} (c);
    \path[->,dashed] (d) edge[reflexive right] node[right]{$b$} (d);
    \path[->,dashed] (d) edge node[below right]{$a$} (e);
    \path[->,dashed] (e) edge[reflexive right] node[right]{$a,b$} (e);
  \end{tikzpicture}
  \caption{A lasso automaton accepting $\{(a^k,ba^j)\mid k,j\in \mathbb{N}\}$.}
  \label{fig:figure1}
\end{figure}
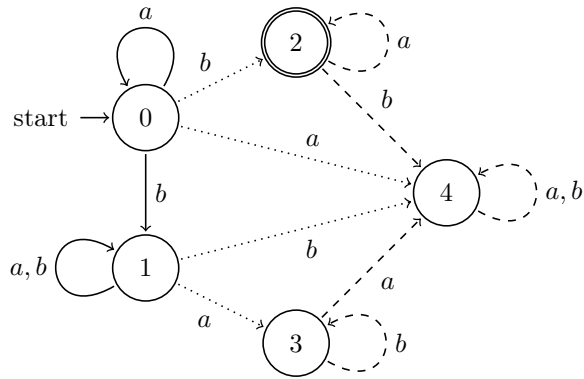

\begin{eg}\label{eg:lassoAutomaton}
  Figure \ref{fig:figure1} depicts a lasso automaton. The spoke states
  are labelled $0$ and $1$. The states $2,3$ and $4$ are loop states.
  Spoke transitions are drawn as solid arrows, switch transitions as
  dotted arrows and loop transitions as dashed arrows.
  The initial state is indicated by the `start' arrow and accepting states
  are drawn as double circles. We give an example run for the lasso
  $(aaa,baa)$. Reading this lasso from the initial state leads to the
  state
   \[
     (\delta_2\col\delta_3)(\delta_1(0,aaa),baa) = \delta_3(\delta_2(\delta_1(0,aaa),b),aa) = \delta_3(\delta_2(0,b),aa) = \delta_3(2,aa) = 2.
  \] 
  Hence $(aaa,baa)$ is accepted by the lasso automaton. The regular lasso
  language it accepts is $\{(a^k,ba^j)\mid k,j\in \mathbb{N}\}$.
\end{eg}

In \cite{ciancia:2019:omegaAutomata}, the authors define an
$\Omega$-automaton as a lasso automaton with special structural
properties. We give an alternative definition of these automata
using saturation.

\begin{defn}\label{defn:omegaAutomaton}
	A lasso automaton $\mathcal{A}$ is \emph{saturated} if for any two
  $\gamma$-equivalent lassos
	$(u_1,v_1),(u_2,v_2)\in\Sigma^{\ast+}: (u_1,v_1)\in L_\circ(\mathcal{A}) \iff (u_2,v_2)\in L_\circ(\mathcal{A})$.
  An \emph{$\Omega$-automaton} is a saturated lasso automaton.
\end{defn}

Finite $\Omega$-automata act as acceptors of regular
$\omega$-languages.  A finite $\Omega$-automaton $\mathcal{A}$ accepts
the regular $\omega$-language $L$ if
$L_\circ(\mathcal{A})=\{(u,v)\mid uv^\omega\in L\}$
\cite{ciancia:2019:omegaAutomata}. Note that $L_\circ(\mathcal{A})$
is always $\sim_\gamma$-saturated (that is, a partition of $\sim_\gamma$ equivalence
classes) for an $\Omega$-automaton $\mathcal{A}$, and that
$\{(u,v)\mid uv^\omega\in L\}=\{(u,v)\mid uv^\omega\in K\}$ implies
$L=K$ for regular $\omega$-languages $K,L$
\cite{calbrix:1994:ultimatelyPeriodicWords}.  The \emph{regular
  $\omega$-language accepted by a finite $\Omega$-automaton
  $\mathcal{A}$} is denoted $L_\omega(\mathcal{A})$.

\begin{figure}[htpb]
  \centering
  \begin{center}
    \begin{tikzpicture}[modal]
      \node[state,accepting] at (0,0) (a) [label=above:{}] {};
      \node[state] at (3,0) (c) [label=above:{}] {};
      \node[state,initial] at (1.5,2) (b) [label=above:{}] {};
      \path[->,dashed] (c) edge[reflexive right] node[right]{$a,b$} (c);
      \path[->] (b) edge[reflexive right] node[right]{$a,b$} (b);
      \path[->,dotted] (b) edge node[above left]{$a$} (a);
      \path[->,dotted] (b) edge node[]{$b$} (c);
      \path[->,dashed] (a) edge[reflexive left] node[]{$a$} (a);
      \path[->,dashed] (a) edge node[below]{$b$} (c);
    \end{tikzpicture}
  \end{center}
  \caption{A saturated lasso automaton ($\Omega$-automaton).}
  \label{fig:figure2}
\end{figure}
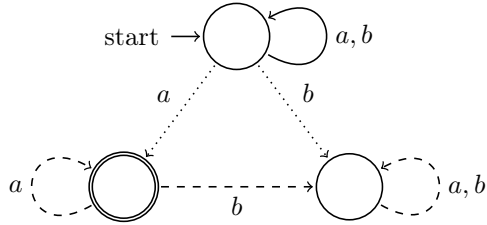

\begin{eg}
  The lasso automaton shown in Figure \ref{fig:figure2} is saturated. It accepts
  the regular lasso language $\{(u,a^k)\mid u\in\Sigma^\ast,k\geq 1\}$ and the
  regular $\omega$-language $\{ua^\omega\mid u\in\Sigma^\ast\} $.
  The automaton from Example \ref{eg:lassoAutomaton} (c.f. Figure \ref{fig:figure1})
  is not saturated as it accepts the lasso $(\epsilon,b)$ but it
  does not accept $(b,b)$ although these two lassos are
  $\gamma$-equivalent.
\end{eg}

\section{Rational Lasso Expressions}\label{sec:3}

In this section we introduce rational lasso expressions, languages and an algebra for
these, which we show to be sound.  A lasso language is \emph{rational}
if it is obtained from rational languages using the operations
\[
U^\circ = \{(\epsilon,u)\mid u\in U\},\qquad U\cdot K = \{(uv,w)\mid u\in U, (v,w)\in K\},\qquad K_1\cup K_2,
\]
where $U$ is a rational language and $K,K_1,K_2$ are rational lasso
languages.

From here on, we assume $\text{RA}$ to be an arbitrary but fixed
algebra of rational expressions of signature $(0,1,+,\cdot,^\ast)$
(for instance $\text{KA}$, the theory of Kleene Algebra \cite{kozen:1994:completeness}).
We write $\vdash_{\mathbf{RA}} t=r$ if $t=r$ is deducible in $\text{RA}$. As it is
always clear from context, we drop the subscript and just write
$\vdash t=r$, sometimes even dropping the turnstyle altogether.
We write $t\leq r$ for
$\vdash t+r=r$, as the $+$-reduct of an $\text{RA}$-algebra is a
join-semilattice. Finally, for a formula $\phi$, the \emph{Iverson bracket} of
$\phi$ is defined as
\[
  [\phi] = \begin{cases}
    1 & \text{if }\phi \text{ is true},\\
    0 & \text{otherwise}.
  \end{cases}
\]

\begin{defn}
  Let $t,r\in \text{Exp}$ with $r\not\in N$. The set
  $\text{Exp}_\circ$ of \emph{rational lasso expressions} is defined
  by the grammar
  \[
  \rho,\sigma ::= 0_\circ \mid t\cdot_\circ \rho\mid \rho +_\circ \sigma \mid r^\circ.
  \]
\end{defn}

Whenever it is clear from context, we drop the subscript $(-)_\circ$, simply
writing $0,+$ and $\cdot$.
We associate a rational lasso language to each rational lasso
expression using the operations defined at the start of this section.

\begin{defn}
  The \emph{language semantics for rational lasso expressions}
  $\llbracket -\rrbracket_\circ:\text{Exp}_\circ\to 2^{\Sigma^{\ast+}}$ is given
  by
  \[\llbracket 0\rrbracket_\circ = \emptyset, \qquad \llbracket t\cdot \rho\rrbracket_\circ = \llbracket t\rrbracket\cdot \llbracket \rho\rrbracket_\circ,\qquad \llbracket \rho+\sigma\rrbracket_\circ = \llbracket \rho\rrbracket_\circ\cup \llbracket \sigma\rrbracket_\circ,\qquad \llbracket r^\circ\rrbracket_\circ = \llbracket r\rrbracket^\circ.\]
  By definition, the lasso language semantics
  $\llbracket -\rrbracket_\circ$ extends the language semantics of
  rational expressions $\llbracket -\rrbracket$.
\end{defn}

\begin{figure}[htpb]
  \centering
  \begin{tikzpicture}[modal]
    \node[state,initial] at (0,3) (a) [label=above:{}] {};
    \node[state] at (3,3) (b) [label=above:{}] {};
    \node[state] at (1.5,1.5) (d) [label=above:{}] {};
    \node[state,accepting] at (3.5,0) (e) [label=above:{}] {};
    \node[state] at (1.5,0) (f) [label=above:{}] {};
    \path[->] (a) edge node[]{$b$} (b);
    \path[->] (a) edge node[below left]{$a$} (d);
    \path[->,dotted] (a) edge[bend right] node[below left]{$a,b$} (f);
    \path[->] (b) edge[reflexive right] node[right]{$a$} (b);
    \path[->] (b) edge node[below right]{$b$} (d);
    \path[->,dotted] (b) edge[bend left] node[above right]{$b$} (e);
    \path[->,dotted] (b) edge[bend left] node[below right]{$a$} (f);
    \path[->] (d) edge[reflexive above] node[]{$a,b$} (d);
    \path[->,dotted] (d) edge node[right]{$a,b$} (f);
    \path[->,dashed] (e) edge node[]{$a,b$} (f);
    \path[->,dashed] (f) edge[reflexive left] node[]{$a,b$} (f);
  \end{tikzpicture}
  \caption{A finite lasso automaton accepting $\llbracket b(a^\ast b^\circ)\rrbracket_\circ$.}
  \label{fig:Figure3}
\end{figure}
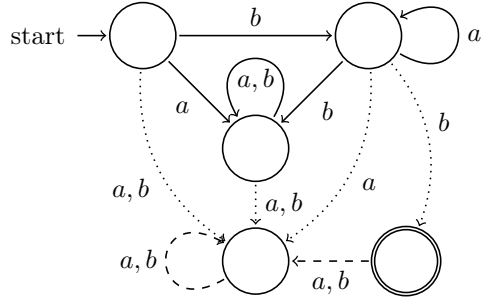

\begin{eg}
  Let $\Sigma=\{a,b\} $ and consider the rational lasso expression
  $b(a^\ast b^\circ)$. We compute its language semantics recursively:
  \[
    \llbracket b(a^\ast b^\circ)\rrbracket_\circ 
    = \llbracket b\rrbracket (\llbracket a^\ast\rrbracket\llbracket b^\circ\rrbracket_\circ) 
    = \{b\} ( \{a^k\mid k\geq 0\} \{(\epsilon,b)\} )
    = \{b\} \{(a^k,b)\mid k\geq 0\}
    = \{(ba^k,b)\mid k\geq 0\}.
  \]
  The rational lasso language obtained is also regular as it is accepted
  by the finite lasso automaton in Figure \ref{fig:Figure3}.
\end{eg}

Next, we introduce a theory to reason about regular lasso
expressions. This theory is sound with respect to the lasso language
semantics, which we require for the construction of a
Brzozowski lasso automaton.

\begin{defn}\label{defn:LAAxioms}
  The \emph{two-sorted theory LA of lasso algebras} extends the theory
  $\text{RA}$ by the following axioms:
  \begin{align*}
    1\cdot \rho &= \rho & 0\cdot \rho &= 0 & \rho+\sigma &= \sigma+\rho \\
    (t+r)\cdot \rho &= t\cdot \rho+ r\cdot \rho & 0^\circ &= 0 & (\rho+\sigma)+\tau &= \rho+(\sigma+\tau)\\
    t\cdot (\rho+\sigma) &= t\cdot \rho + t\cdot \sigma & 0 + \rho &= \rho & \rho+\rho &= \rho \\
    t\cdot (r\cdot \rho) &= (t\cdot r)\cdot \rho & t\cdot 0 &= 0 & (t+r)^\circ &= t^\circ + r^\circ\ (t,r\not\in N)
  \end{align*}
  with $t,r\in \text{Exp}$ and
  $\rho,\sigma,\tau\in \text{Exp}_\circ$.
  The axioms together with the laws for equality and the substitution
  of provably equivalent rational and rational lasso expressions gives
  us the deductive system $\mathbf{LA}$. We write
  $\vdash_{\mathbf{LA}}\rho=\sigma$ when
  the equation $\rho=\sigma$ is deducible in
  LA. Whenever it is clear from context, we drop the turnstile
  $\vdash_{\mathbf{LA}}$.
\end{defn}

The next proposition shows that the theory of lasso algebras is sound
with respect to the language semantics for rational lasso expressions.
We make no claim about its completeness.

\begin{prop}[Soundness]\label{prop:soundness}
  Let $\rho,\sigma\in \text{Exp}_\circ$. Then
  $\vdash_{\mathbf{LA}}\rho=\sigma \implies \llbracket\rho\rrbracket_\circ = \llbracket\sigma\rrbracket_\circ$.
\end{prop}

  The verification of the laws in Definition \ref{defn:LAAxioms} is routine.
  It relies on the following properties of the rational lasso operations.
  For rational languages $U,V$ and rational lasso languages
  $K,L$, we have
  \begin{enumerate}
    \item $U\cdot (V\cdot L) = (U\cdot V)\cdot L$,
    \item $(U\cup V)\cdot L = (U\cdot L)\cup (V\cdot L)$,
    \item $U\cdot (K\cup L)=U\cdot K \cup U\cdot L$,
    \item $(U \cup V)^\circ = U^\circ \cup V^\circ$.
  \end{enumerate}

\begin{rmk}\label{rmk:wagner}
  We briefly highlight the differences between lasso and Wagner
  algebras (\cite{wagner:1976:axiomatisierung}). A Wagner algebra is a
  two-sorted algebra similar to the lasso algebra but with an
  operation $(-)^\omega$ instead of $(-)^\circ$. They are used to
  reason about rational $\omega$-expressions and Wagner showed
  completeness of his axiomatisation with respect to the language
  semantics for rational $\omega$-expressions. Wagner's axiomatisation
  looks very similar to that of a lasso algebra.  However, the unary
  operations $(-)^\circ$ and $(-)^\omega$ satisfy different laws, the
  ones for $(-)^\omega$ being
  \[
    (t\cdot r)^\omega = t\cdot (r\cdot t)^\omega \qquad \text{ and }\qquad t^\omega = (t^+)^\omega.
  \]
  Other than this, there are two more subtle differences that can be pointed
  out:
  \begin{enumerate}
  \item from the $(-)^\omega$-axioms, one can deduce that
    $0^\omega=0$, which is not the case for lasso algebras (i.e. we
    need the axiom $0^\circ=0$),
    \item Wagner has an additional derivation rule, which allows to solve
      equations of a particular type. Such a rule is not given for lasso
      algebras.
  \end{enumerate}
\end{rmk}

Analogously to the situation for rational and rational $\omega$-expressions,
each rational lasso expression is provably equivalent to a rational lasso expression
of the form $\sum_{i=1}^n t_i\cdot r_i^\circ$. Such a form is called
a \emph{disjunctive form}, but note that a disjunctive form is not unique.

\begin{prop}\label{prop:normalForm}
  Let $\rho\in\text{Exp}_\circ$. Then there exists $n\in \mathbb{N}$ and
  $t_1,\ldots,t_n,r_1,\ldots,r_n\in \text{Exp}$ with
  $r_i\not\in N$ for all $1\leq i\leq n$ such that
  \[
    \vdash_{\mathbf{LA}} \rho = \sum_{i=1}^n t_i\cdot r_i^\circ.
  \]
\end{prop}

\begin{proof}
  We show this by structural induction on rational lasso expressions.
  For $0$ we have that $0 = 0\cdot 0^\circ$. 
  For  $t^\circ$ we have  $t^\circ = 1\cdot t^\circ$.
  For $r\cdot \rho$, we know by the induction
  hypothesis that we can find a finite number of pairs $(t_i,r_i)$ such that
  $\rho= \sum_i (t_i,r_i)$. It follows that
  \[
    r\cdot \rho = r\cdot \sum_i t_i\cdot r_i^\circ = \sum_i r\cdot (t_i\cdot r_i^\circ) = \sum_i (r\cdot t_i)\cdot r_i^\circ.
  \] 
  Finally, for $\rho+\sigma$ we can find two finite
  sets $\{(t_i,r_i)\}_{i\in I},\{(t_j,r_j)\}_{j\in J}$ with
  $\rho= \sum_i (t_i,r_i)$ and
  $\sigma= \sum_j (t_j,r_j)$, then clearly
  $\rho+\sigma = \sum_{k\in I\cup J} (t_k,r_k)$ as
  required.
\end{proof}

We often write $(t,r)$ instead of $t\cdot r^\circ$, which
accentuates the distinction between the `finite' and `infinite' part more.
It also allows for a more direct correspondence to the language semantics as
\[
\llbracket(t,r)\rrbracket_\circ = \{(u,v)\in\Sigma^{\ast +}\mid u\in \llbracket t\rrbracket, v\in\llbracket r\rrbracket\}.
\]

\section{Rational Lasso Languages are Regular}\label{sec:4}

In this section we explore a construction for lasso automata from rational
lasso expressions. Given a rational lasso expression $\rho$, our aim
is to build a finite lasso automaton $\mathcal{A}$ which accepts
$\llbracket\rho\rrbracket_\circ$. This shows one direction of Kleene's
Theorem, namely that every rational lasso language is regular.

We start the section by going over the standard Brzozowski construction for deterministic
finite automata (\cite{brzozowski:1964:derivativesRE}). The rest of the section
generalises this construction to finite lasso automata. For this we first define
Brzozowski derivatives for the spoke and switch transitions. We then show some results
on how they interact with the theory of lasso algebras and prove a fundamental
theorem. This allows us to define a Brzozowski lasso automaton whose spoke
states are rational lasso expressions with the property that the language
accepted from an initial state $\rho\in\text{Exp}_\circ$ is $\llbracket \rho\rrbracket_\circ$.
Finally, we introduce a suitable equivalence relation to quotient our automaton
which guarantees that our final lasso automaton has only finitely many
reachable states from any equivalence class of rational lasso expressions.

If we consider a word $a_1\ldots a_n$, then after reading $a_1$ what remains
is $a_2\ldots a_n$. This is called a left-quotient and forms the basic idea
behind the Brzozowski derivative. The idea can be taken to the level
of sets where we start with a language $U$ and wonder what remains if we tried
to read an $a$ from all the words in $U$, again this is a left quotient often
written $a^{-1}U=\{u\mid au\in U\} $. The Brzozowski derivative tries to capture this notion
of a left-quotient syntactically.

The next set of definitions and results are well-known (\cite{brzozowski:1964:derivativesRE}).
The definition of the Brzozowski derivative makes use of the Iverson bracket
introduced in Section \ref{sec:3}.

\begin{defn}
  The \emph{Brzozowski derivative} $d:\text{Exp}\to \text{Exp}^\Sigma$ is defined recursively
  on rational expressions.
  \begin{align*}
    d(0,a)&=0 & d(t+r,a)&=d(t,a)+d(r,a) & d(b,a)&=[b=a]\\
    d(1,a)&=0	& d(t\cdot r,a)&=d(t,a)\cdot r+[t\in N]\cdot d(r,a) & d(t^\ast,a)&=d(t,a)\cdot t^\ast
  \end{align*}
\end{defn}

The next propositions show that the Brzozowski derivative does indeed
capture a left-quotient syntactically. Firstly, if two rational
languages $U$ and $V$ are equal, so should be their quotients
$a^{-1}U=a^{-1}V$ for all $a\in\Sigma$. Secondly, we expect that
$a(a^{-1}U)\subseteq U$ for all $a\in \Sigma$, and moreover, that
$U=\bigcup_{a\in \Sigma} a(a^{-1}U)$ modulo the empty word.

\begin{prop}
  The Brzozowski derivative preserves provable equality of terms.
  \[
  \forall t,r\in \text{Exp}:\ \vdash_{\mathbf{RA}}t=r \implies \left(\forall a\in\Sigma:\ \vdash_{\mathbf{RA}}d(t,a)=d(r,a)\right).
\]
\end{prop}

\begin{prop}[Fundamental Theorem]\label{prop:fundamentalTheorem}
  If $t\in \text{Exp}$. Then
  \[
    \vdash t = [t\in N] + \sum_{a\in \Sigma} a\cdot d(t,a).
  \] 
\end{prop}

The Brzozowski derivative together with the set $N$ of expressions having the
empty word property allows us to define the
\emph{deterministic Brzozowski automaton} $\mathcal{B}=(\text{Exp},d,N)$.
We do not settle on an initial state, but join it later depending on which language
we would like the Brzozowski automaton to accept. In fact, the next proposition
shows that if we choose $t\in \text{Exp}$ as initial state, then the Brzozowski
automaton is going to accept the language $\llbracket t\rrbracket$.

\begin{prop}
  For all $t\in \text{Exp}$ we have $L(\mathcal{B},t)=\llbracket t\rrbracket$.
\end{prop}

Unfortunately, the Brzozowski automaton is not necessarily a DFA. If we choose
$t\in \text{Exp}$ as initial state, it could happen that there are infinitely
many reachable states. To solve this we introduce an equivalence relation
by which we can quotient the Brzozowski automaton. We define
$\sim_B\subseteq \text{Exp}^2$ to be the least equivalence relation satisfying:
\begin{align*}
  1\cdot t &\sim_B t & 0\cdot t &\sim_B 0 & t &\sim_B t+t & t+r &\sim_B r+t & (t+r)+g &\sim_B t+(r+g).
\end{align*}
We write $\sim$ whenever this does not lead to confusion. Now we can quotient
the state space, but this also means that we have to modify the Brzozowski
derivative and our set of accepting states. This is done by defining the
derivative of an equivalence class to be the equivalence class of the
derivative, and to make an equivalence class accepting if one of its members
has the empty word property. The next lemma shows that this is well-defined.

\begin{lem}\label{lem:BrzozowskiDerivativeEWPCompatibleWithTildeB}
  The equivalence relation $\sim$ is compatible with the Brzozowski derivative $d$ and
  with the predicate $N$, i.e. if $t,r\in \text{Exp}$ and $t\sim r$, then
  \begin{enumerate}
    \item $t\in N \iff r\in N$ and
    \item $\forall a\in \Sigma: d(t,a)\sim d(r,a)$.
  \end{enumerate}
\end{lem}

We let $\widehat{d}:{\text{Exp}/_{\sim}}\to (\text{Exp}/_{\sim})^\Sigma$ be
the map given by $\widehat{d}([t]_\sim,a)=[d(t,a)]_\sim$ and define the predicate
$\widehat{N}\subseteq (\text{Exp}/_{\sim})^2$ by $[t]_\sim\in \widehat{N}\iff t\in N$.

\begin{cor}
  The derivative $\widehat{d}$ and the predicate $\widehat{N}$ are well-defined.
\end{cor}

This concludes the construction of the quotiented automaton
$\widehat{\mathcal{B}}=(\text{Exp}/{\sim},\widehat{d},\widehat{N})$. The next
result shows that if we choose $[t]_\sim$ as initial state, then
we obtain a DFA which accepts $\llbracket t\rrbracket$.

\begin{thm}
  For all $t\in \text{Exp}$, $L(\widehat{\mathcal{B}},[t]_\sim)=\llbracket t\rrbracket$ and
  the set of states reachable from $[t]_\sim$ is finite.
\end{thm}

The next example shows that with minimal effort and using the Brzozowski
construction, we can quickly come up with a lasso automaton for a simple
rational lasso expression.

\begin{eg}\label{eg:BrzozowskiConstruction}
  Let $\rho=(b(ab)^\ast,ab^\ast)$. Then $(u,v)\in\llbracket \rho\rrbracket_\circ$
  if and only if $u\in\llbracket b(ab)^\ast\rrbracket$ and $v\in\llbracket ab^\ast\rrbracket$.
  So the idea is to build one \textcolor{blue}{DFA} for $b(ab)^\ast$ and one
  \textcolor{red}{DFA} for $ab^\ast$, which
  correspond to the spoke and loop part of the lasso automaton, and then link
  them. The construction of the DFAs is done using Brzozowski derivatives,
  which yields the following DFAs:
  
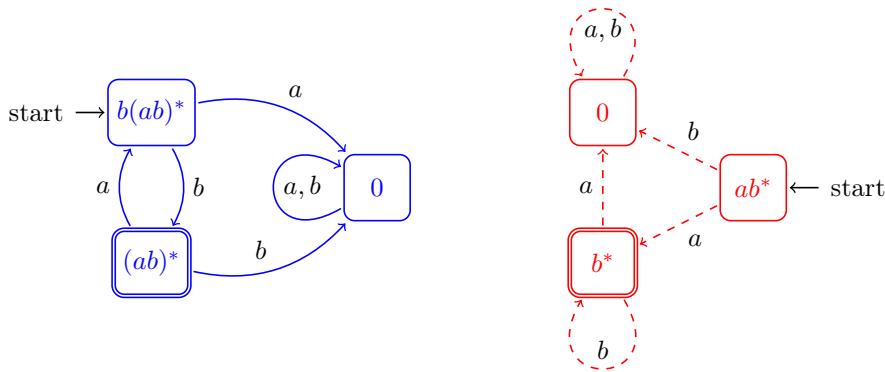
\begin{center}
    \begin{tikzpicture}[modal]
      \node[state,rectangle,rounded corners,initial,blue] at (0,2) (a) [label=above:{}] {$b(ab)^\ast$};
      \node[state,rectangle,rounded corners,accepting,blue] at (0,0) (b) [label=above:{}] {$(ab)^\ast$};
      \node[state,rectangle,rounded corners,blue] at (3,1) (c) [label=above:{}] {$0$};
      \node[state,rectangle,rounded corners,red] at (6,2) (d) [label=above:{}] {$0$};
      \node[state,rectangle,rounded corners,accepting,red] at (6,0) (e) [label=above:{}] {$b^\ast$};
      \node[state,rectangle,rounded corners,initial,initial where=right,red] at (8,1) (f) [label=above:{}] {$ab^\ast$};
      \path[->,blue] (a) edge[bend left] node[black]{$b$} (b);
      \path[->,blue] (a) edge[bend left] node[black]{$a$} (c);
      \path[->,blue] (b) edge[bend left] node[black]{$a$} (a);
      \path[->,blue] (b) edge[bend right] node[black]{$b$} (c);
      \path[->,blue] (c) edge[reflexive left] node[right,black]{$a,b$} (c);
      \path[->,dashed,red] (d) edge[reflexive above] node[below,black]{$a,b$} (d);
      \path[->,dashed,red] (e) edge node[black]{$a$} (d);
      \path[->,dashed,red] (e) edge[reflexive below] node[above,black]{$b$} (e);
      \path[->,dashed,red] (f) edge node[black,above right]{$b$} (d);
      \path[->,dashed,red] (f) edge node[black]{$a$} (e);
    \end{tikzpicture}
  \end{center}
  Note that we are allowed to transition from the first to the second
  DFA immediately after we have read $u$, that is, when we
  have reached an accepting state in the \textcolor{blue}{spoke
    DFA}. In order to determine where to switch to, we think of our
  spoke state as the initial state of the loop
  DFA.  From any other state in the spoke DFA (i.e.\thinspace the
  non-accepting ones), attempting to transition just leads to a dead
  state. In the final lasso automaton the initial state of the second
  DFA is omitted as it is not reachable.
  
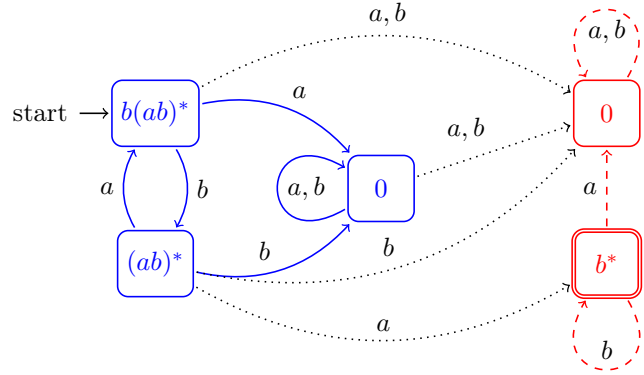
\begin{center}
    \begin{tikzpicture}[modal]
      \node[state,rectangle,rounded corners,initial,blue] at (0,2) (a) [label=above:{}] {$b(ab)^\ast$};
      \node[state,rectangle,rounded corners,blue] at (0,0) (b) [label=above:{}] {$(ab)^\ast$};
      \node[state,rectangle,rounded corners,blue] at (3,1) (c) [label=above:{}] {$0$};
      \node[state,rectangle,rounded corners,red] at (6,2) (d) [label=above:{}] {$0$};
      \node[state,rectangle,rounded corners,accepting,red] at (6,0) (e) [label=above:{}] {$b^\ast$};
      \path[->,blue] (a) edge[bend left] node[black]{$b$} (b);
      \path[->,blue] (a) edge[bend left] node[black]{$a$} (c);
      \path[->,dotted] (a) edge[bend left] node[black]{$a,b$} (d);
      \path[->,blue] (b) edge[bend left] node[black]{$a$} (a);
      \path[->,blue] (b) edge[bend right] node[black]{$b$} (c);
      \path[->,dotted] (b) edge[bend right] node[black]{$b$} (d);
      \path[->,dotted] (b) edge[bend right] node[black]{$a$} (e);
      \path[->,blue] (c) edge[reflexive left] node[right,black]{$a,b$} (c);
      \path[->,dotted] (c) edge node[black]{$a,b$} (d);
      \path[->,dashed,red] (d) edge[reflexive above] node[below,black]{$a,b$} (d);
      \path[->,dashed,red] (e) edge node[black]{$a$} (d);
      \path[->,dashed,red] (e) edge[reflexive below] node[above,black]{$b$} (e);
    \end{tikzpicture}
  \end{center}
\end{eg}

We turn this intuition into a formal construction, first defining the derivatives and
showing some of their properties, and then proving that the lasso automaton
we obtain accepts the desired language.

\begin{defn}
  Define the \emph{spoke} and \emph{switch Brzozowski derivatives}
  $d_1\colon\text{Exp}_\circ \to \text{Exp}_\circ^\Sigma$ and
  $d_2\colon\text{Exp}_\circ\to \text{Exp}^\Sigma$ recursively on the structure of
  rational lasso expressions as
  {\allowdisplaybreaks
  \begin{align*}
    d_1(0,a)&=0 & d_2(0,a)&=0 \\
    d_1(t^\circ,a)&=0 & d_2(t^\circ,a)&=d(t,a)\\
    d_1(\rho+\sigma,a)&=d_1(\rho,a)+d_1(\sigma,a) & d_2(\rho+\sigma,a)&=d_2(\rho,a)+d_2(\sigma,a) \\
    d_1(r\cdot \rho,a)&=d(r,a)\cdot \rho + [r\in N]\cdot d_1(\rho,a) & d_2(r\cdot \rho,a)&=[r\in N]\cdot d_2(\rho,a)
       \end{align*}
  }
\end{defn}

Analogously to the situation with the classical Brzozowski construction,
we can show that the additional derivatives preserve provable equality of
terms.

\begin{prop}\label{prop:derivativesPreserveProvableEquality}
  Let $\rho,\sigma\in\text{Exp}_\circ$ with $\vdash_{\mathbf{LA}}\rho=\sigma$.
  Then for all $a\in \Sigma$:
  \begin{enumerate}
    \item $\vdash_{\mathbf{RA}} d_2(\rho,a)=d_2(\sigma,a)$ and
    \item $\vdash_{\mathbf{LA}}d_1(\rho,a)=d_1(\sigma,a)$.
  \end{enumerate}
\end{prop}

\begin{proof}
  We need to check all
  the relevant equations from our theory and make sure that it works
  for substitution.  We only treat some of the equations from the
  theory as this involves simple manipulations of equations.  The
  checks for substitution are readily verified and hence omitted.

  For both derivatives we show the claim for the equations
  $t\cdot (r\cdot \rho)=(t\cdot r)\cdot \rho$ and $(t+r)^\circ = t^\circ + r^\circ$.
  When treating the first equation, we need this additional fact:
  \[
    \vdash[t\in N]\cdot [r\in N]=[t\cdot r\in N].
  \]

  We start with the derivative $d_1$ and check the first equation:
  {\allowdisplaybreaks
  \begin{align*}
    d_1(t\cdot (r\cdot \rho),a) &= d(t,a)\cdot (r\cdot \rho) + [t\in N]\cdot d_1(r\cdot \rho,a) && (\text{defn. $d_1$})\\
                                         &= d(t,a)\cdot (r\cdot \rho) + [t\in N]\cdot (d(r,a)\cdot \rho \\
                                         &\qquad +[r\in N]\cdot d_1(\rho,a)) && (\text{defn. $d_1$}) \\
                                         &= d(t,a)\cdot (r\cdot \rho) + [t\in N]\cdot (d(r,a)\cdot \rho) \\
                                         &\qquad +[t\in N]\cdot ([r\in N]\cdot d_1(\rho,a)) && (\text{dist.}) \\
                                         &= (d(t,a)\cdot r)\cdot \rho + ([t\in N]\cdot d(r,a))\cdot \rho \\
                                         &\qquad +([t\in N]\cdot [r\in N])\cdot d_1(\rho,a) && (\text{mixed assoc.})\\
                                         &= \left(d(t,a)\cdot r + [t\in N]\cdot d(r,a)\right)\cdot \rho\\
                                         &\qquad +[t\cdot r\in N]\cdot d_1(\rho,a) && (\text{dist. \& fact}) \\
                                         &= d(t\cdot r,a)\cdot \rho+[t\cdot r\in N]\cdot d_1(\rho,a) && (\text{defn. $d$}) \\
                                         &= d_1((t\cdot r)\cdot \rho,a) && (\text{defn. $d_1$})
  \end{align*}
  }
  For the other equation we get
  \[
    d_1((t+r)^\circ,a) = 0 = 0+0 = d_1(t^\circ,a)+d_1(r^\circ,a)=d_1(t^\circ+r^\circ,a),
  \] using the definition of $d_1$ and $\vdash_{\mathbf{LA}}0+0=0$.
  
  Next we turn to $d_2$. Again, for the first equation we have
  \begin{align*}
    d_2(t\cdot (r\cdot \rho),a) &= [t\in N]\cdot d_2(r\cdot \rho,a) && (\text{defn. $d_2$})\\
                                         &= [t\in N]\cdot \left([r\in N]\cdot d_2(\rho,a)\right) && (\text{defn. $d_2$}) \\
                                         &= \left([t\in N]\cdot [r\in N]\right)\cdot d_2(\rho,a) && (\text{mixed assoc.})\\
                                         &= [t\cdot r\in N]\cdot d_2(\rho,a) && (\text{fact})\\
                                         &= d_2((t\cdot r)\cdot \rho,a). && (\text{defn. $d_2$})
  \end{align*}
  Finally, for the second equation we obtain
  \begin{align*}
    d_2((t+r)^\circ,a) &= d(t+r,a) && (\text{defn. $d_2$}) \\
                            &= d(t,a)+d(r,a) && (\text{defn. $d$}) \\
                            &= d_2(t^\circ,a)+d_2(r^\circ,a) && (\text{defn. $d_2$}) \\
                            &= d_2(t^\circ+r^\circ,a). && (\text{defn. $d_2$})
  \end{align*}
  The other proofs are similar.
\end{proof}

Next we show a fundamental theorem. The two-sorted nature of lasso automata
becomes clearly visible with a rational lasso expression being decomposed
in terms of the two derivatives $d_1$ and $d_2$.

\begin{prop}[Fundamental Theorem]\label{prop:fundamentalTheoremLasso}
  Let $\rho\in \text{Exp}_\circ$. Then
  \[
    \vdash_{\mathbf{LA}}\rho = \left(\sum_{a\in \Sigma}a\cdot d_1(\rho,a)\right) + \left( \sum_{a\in \Sigma} a\cdot d_2(\rho,a) \right)^\circ. 
  \] 
\end{prop}

\begin{proof}
  We proceed by structural induction on rational lasso expressions. For the case
  where $\rho=0$ we apply the definition of $d_1$ and $d_2$, and use the
  fact that $0$ is an absorbing element.
  \[
      \sum_{a\in \Sigma} a\cdot d_1(0,a) + \left( \sum_{a\in\Sigma} a\cdot d_2(0,a)\right)^\circ = \sum_a a\cdot 0 + \left(\sum_a a\cdot 0 \right)^\circ = 0
    \]
  For the case $\rho=t^0$, we apply the definition of the derivatives, use the
  fact that $0$ is an absorbing element, and use the Fundamental Theorem
  (Proposition \ref{prop:fundamentalTheorem}) for rational expressions with the fact that $t\not\in N$.
  \[
      \sum_{a\in \Sigma} a\cdot d_1(t^\circ,a) + \left( \sum_{a\in\Sigma} a\cdot d_2(t^\circ,a)\right)^\circ = \sum_a a\cdot 0 + \left(\sum_a a\cdot d(t,a) \right)^\circ = t^\circ
    \]
  For the case of $\rho+\sigma$, we make use of the induction hypothesis,
  distributivity, associativity, commutativity and the fact that $(-)^\circ$ 
  distributes over sums.
    \begin{align*}
      &\phantom{km} \sum_{a\in \Sigma} a\cdot d_1(\rho+\sigma,a) + \left( \sum_{a\in\Sigma} a\cdot d_2(\rho+\sigma,a)\right)^\circ \\
      &= \sum_a a\cdot d_1(\rho,a) + \sum_a a\cdot d_1(\sigma,a) + \left(\sum_a a\cdot d_2(\rho,a)+\sum_a a\cdot d_2(\sigma,a) \right)^\circ \\
      &= \sum_a a\cdot d_1(\rho,a) + \sum_a a\cdot d_1(\sigma,a) + \left(\sum_a a\cdot d_2(\rho,a)\right)^\circ +\left(\sum_a a\cdot d_2(\sigma,a) \right)^\circ \\
      &= \sum_a a\cdot d_1(\rho,a) + \left(\sum_a a\cdot d_2(\rho,a)\right)^\circ + \sum_a a\cdot d_1(\sigma,a) + \left(\sum_a a\cdot d_2(\sigma,a) \right)^\circ \\
      &= \rho + \sigma
    \end{align*}
  Finally, for the case $r\cdot \rho$ we use the Fundamental Theorem for rational
  expressions (Proposition \ref{prop:fundamentalTheorem}) on $r$, make use of
  the induction hypothesis, distributivity, associativity, commutativity and
  the additional facts
  \[
    \vdash t\cdot [r\in N]=[r\in N]\cdot t \qquad \text{ and }\qquad
    \vdash_{\mathbf{LA}} ([r\in N]\cdot t)^\circ = [r\in N]\cdot t^\circ.
  \]
  This yields:
  \allowdisplaybreaks{
    \begin{align*}
      &\phantom{km} \sum_{a\in \Sigma} a\cdot d_1(r\cdot \rho,a) + \left( \sum_{a\in\Sigma} a\cdot d_2(r\cdot \rho,a)\right)^\circ \\
      &= \sum_{a} a\cdot (d(r,a)\cdot \rho + [r\in N]\cdot d_1(\rho,a)) + \left( \sum_{a} a\cdot [r\in N] \cdot d_2(\rho,a)\right)^\circ \\
      &= \sum_{a} a\cdot (d(r,a)\cdot \rho) +\sum_a a\cdot \left([r\in N]\cdot d_1(\rho,a)\right) + \left( \sum_{a\in\Sigma} a\cdot [r\in N] \cdot d_2(\rho,a)\right)^\circ \\
      &= \sum_{a} (a\cdot d(r,a))\cdot \rho +[r\in N]\cdot \sum_a a\cdot d_1(\rho,a) + [r\in N]\cdot \left( \sum_{a\in\Sigma} a \cdot d_2(\rho,a)\right)^\circ \\
      &= \left(\sum_{a} a\cdot d(r,a)\right)\cdot \rho +[r\in N]\cdot \left(\sum_a  a\cdot d_1(\rho,a) + \left( \sum_{a\in\Sigma} a \cdot d_2(\rho,a)\right)^\circ\right) \\
      &= \left(\sum_{a} a\cdot d(r,a)\right)\cdot \rho +[r\in N]\cdot \rho \\
      &= \left(\left(\sum_{a} a\cdot d(r,a)\right) +[r\in N]\right)\cdot \rho \\
      &= r\cdot \rho \qedhere
    \end{align*}
  }
\end{proof}

We call the lasso automaton
$\mathcal{C}=(\text{Exp}_\circ,\text{Exp},d_1,d_2,d,N)$ the
\emph{Brzozowski lasso automaton}. The loop states together with the Brzozowski
derivative and the accepting states forms the standard Brzozowski automaton
$\mathcal{B}$, so we can think of $\mathcal{C}$ as extending $\mathcal{B}$.
Moreover, if we set the initial state to be $\rho\in\text{Exp}_\circ$, then
the lasso language accepted by $\mathcal{C}$ is $\llbracket\rho\rrbracket_\circ$.

\begin{prop}\label{prop:BrzozowskiLassoAutomaton}
  If $\rho\in\text{Exp}_\circ$, then $L_\circ(\mathcal{C},\rho)=\llbracket\rho\rrbracket_\circ$.
\end{prop}

\begin{proof}
  We proceed by induction on the length of the spoke word.
  For the base case assume the spoke word to be the empty
  word, i.e. consider the lasso 
  $(\epsilon,au)$. By using
  the Fundamental Theorem (Proposition \ref{prop:fundamentalTheoremLasso}) and
  Soundness (Proposition \ref{prop:soundness}) we get
  \[
    (\epsilon,au)\in \llbracket \rho\rrbracket_\circ \iff (\epsilon,au)\in \left\llbracket \sum_a a\cdot d_1(\rho,a) + \left(\sum_a a\cdot d_2(\rho,a) \right)^\circ \right\rrbracket_\circ.
  \]
  It follows that $au\in \left\llbracket a\cdot d_2(\rho,a) \right\rrbracket$
  which is equivalent to $\epsilon \in \llbracket d(d_2(\rho,a),u) \rrbracket$,
  and in turn to $d(d_2(\rho,a),u) \in N$. By the definition of acceptance
  this is the same as $(\epsilon,au)\in L_\circ(\mathcal{C},\rho)$.
  This establishes the base case.
  For the induction step consider the lasso $(au,v)$. Again by using
  the Fundamental Theorem and Soundness we get that
  \[
    (au,v)\in \llbracket \rho\rrbracket_\circ \iff (au,v)\in \left\llbracket \sum_a a\cdot d_1(\rho,a) + \left(\sum_a a\cdot d_2(\rho,a) \right)^\circ \right\rrbracket_\circ.
  \]
  By the definition of the semantics map this is the same as
  $(au,v)\in \left\llbracket a\cdot d_1(\rho,a) \right\rrbracket_\circ$
  which in turn is equivalent to
  $(u,v)\in \left\llbracket d_1(\rho,a) \right\rrbracket_\circ$. Using the
  induction hypothesis we get 
  $(u,v)\in L_\circ(\mathcal{C},d_1(\rho,a))$ which is equivalent to
  $(au,v)\in L_\circ(\mathcal{C},\rho)$.
\end{proof}

The Brzozowski lasso automaton is not necessarily finite. The easiest
fix to this is to assume that our rational lasso expressions are
disjunctive forms (c.f.\thinspace Proposition \ref{prop:normalForm}).
To accommodate this assumption with regards to our Brzozowski lasso automaton,
we slightly adapt the definition of $d_1$ so that $d_1(\rho,a)$ is again a
disjunctive form for $a\in \Sigma$. This is done by defining
\[
  d_1(t\cdot s^\circ,a) := d(t,a)\cdot s^\circ.
\]

It is important to point out that this does not jeopardise our earlier
result (Proposition \ref{prop:BrzozowskiLassoAutomaton}). Indeed without this change we would have
\[
  d_1(t\cdot s^\circ,a) = d(t,a)\cdot s^\circ + [t\in N]\cdot d(s^\circ,a) = d(t,a)\cdot s^\circ + [t\in N]\cdot 0,
\] which is provably equivalent to $d(t,a)\cdot s^\circ$.
We ignore this technicality in pursuit of a cleaner presentation.

It should also be remarked, that with these changes $d_1$ only acts on the spoke expressions (where
it acts just like a normal Brzozowski derivative),
$d_2$ acts on the loop expression to give a rational expression
(provided that the spoke expression has the empty word property) and finally
$d$ acts on the obtained rational expression. In this sense, we are
really constructing two DAs using Brzozowski derivatives
and linking them appropriately with $d_2$, as in Example \ref{eg:BrzozowskiConstruction}.

In the remainder of this section, we quotient our Brzozowski lasso
automaton by a suitable equivalence relation which respects the
derivatives, and such that for any rational lasso expression, its
set of successors is finite. Using our
assumption that rational lasso expressions are in a disjunctive form,
the definition makes use of the equivalence relation $\sim_B$ which we
introduced earlier as the least equivalence relation satisfying
\begin{align*}
  1\cdot t &\sim_B t & 0\cdot t &\sim_B 0 & t &\sim_B t+t & t+r &\sim_B r+t & (t+r)+g &\sim_B t+(r+g).
\end{align*}
Using the equivalence relation $\sim_B$, we define
${\sim_C}\subseteq \text{Exp}_\circ^2$ to be the least equivalence
relation such that
\[
  \sum_{i=1}^{n}t_i\cdot r_i^\circ \sim_{C} \sum_{i=1}^n t'_i\cdot (r'_i)^\circ \iff \forall \ 1\leq i\leq n: t_i\sim_B t'_i \text{ and } r_i= r'_i.
\] 
As with the standard Brzozowski derivative, we define our derivatives
$d_1$ and $d_2$ now on equivalence classes and then show that this is
well-defined. Let $\widehat{d_1}$ and $\widehat{d_2}$ be the following maps:
\begin{alignat*}{2}
  \widehat{d_1}: \text{Exp}_\circ/{\sim_C}  &\longrightarrow (\text{Exp}_\circ/{\sim_C})^\Sigma & \widehat{d_2}: \text{Exp}_\circ/{\sim_C}  &\longrightarrow (\text{Exp}/{\sim_B})^\Sigma \\
  [\rho]_{\sim_C} &\longmapsto \lambda a\in\Sigma. [d_1(\rho,a)]_{\sim_C} &\qquad [\rho]_{\sim_C} &\longmapsto \lambda a\in\Sigma. [d_2(\rho,a)]_{\sim_B}
\end{alignat*}

\begin{prop}
  The maps $\widehat{d_1}$ and $\widehat{d_2}$ are well defined.
\end{prop}

\begin{proof}
  Let $\sum_{i=1}^n t_i\cdot s_i^\circ \sim_C \sum_{i=1}^n t'_i\cdot (s'_i)^\circ$.
  We have to show that their $d_1$ derivatives are in the same
  $\sim_C$ equivalence class, and that their $d_2$ derivatives are in the
  same $\sim_B$ equivalence class. We start by taking their $d_1$ derivatives
  with respect to some $a\in \Sigma$:
  \[
    d_1\left(\sum_{i=1}^n t_i\cdot s_i^\circ,a\right) = \sum_{i=1}^n d(t_i,a)\cdot s_i^\circ\qquad \text{ and }\qquad
    d_1\left(\sum_{i=1}^n t'_i\cdot (s'_i)^\circ,a\right) = \sum_{i=1}^n d(t'_i,a)\cdot (s'_i)^\circ.
  \]
  By our assumption and the definition of $\sim_C$, it follows that $t_i\sim_B t'_i$.
  Hence  $d(t_i,a)\sim_B d(t'_i,a)$ by Lemma
  \ref{lem:BrzozowskiDerivativeEWPCompatibleWithTildeB}. Furthermore, we also have $s_i=s'_i$ and so
  \[
    d_1\left(\sum_{i=1}^n t_i\cdot s_i^\circ,a\right) \sim_C d_1\left(\sum_{i=1}^n t_i\cdot s_i^\circ,a\right),
  \] that is their $d_1$ derivatives are in the same $\sim_C$ equivalence class.
  Next, we take the $d_2$ derivatives with respect to some $a\in\Sigma$:
  \[
    d_2\left(\sum_{i=1}^n t_i\cdot s_i^\circ,a\right) = \sum_{i=1}[t_i\in N]\cdot d(s_i,a) \qquad \text{ and }\qquad
    d_2\left(\sum_{i=1}^n t'_i\cdot (s'_i)^\circ,a\right) = \sum_{i=1}[t'_i\in N]\cdot d(s'_i,a).
  \]
  Again, by our assumption $t_i\sim_B t'_i$ so that $[t_i\in N]=[t'_i\in N]$ by Lemma
  \ref{lem:BrzozowskiDerivativeEWPCompatibleWithTildeB}. Moreover, $d(s_i,a)= d(s'_i,a)$ follows from
  $s_i=s_i'$, hence
  \[
    d_2\left(\sum_{i=1}^n t_i\cdot s_i^\circ,a\right)=d_2\left(\sum_{i=1}^n t_i'\cdot (s_i')^\circ,a\right)
  \] and so their $d_2$ derivatives are in particular related by $\sim_B$.
\end{proof}

We obtain the lasso automaton 
$\widehat{\mathcal{C}}=(\text{Exp}_\circ/{\sim_C},\text{Exp}/{\sim_B},\widehat{d_1},\widehat{d_2},\widehat{d},\widehat{N})$
by quotienting the lasso automaton $\mathcal{C}$ by $(\sim_C,\sim_B)$.

\begin{thm}\label{thm:BrzozowskiOmegaAutomaton}
  For all $\rho\in \text{Exp}_\circ$, where $\rho$ is a disjunctive form,
  $L_\circ(\widehat{\mathcal{C}},[\rho]_{\sim_C})=\llbracket \rho\rrbracket_\circ$ and
  the set of states reachable from $[\rho]_{\sim_C}$ is finite.
\end{thm}

\begin{proof}
  The claim that $L_\circ(\widehat{\mathcal{C}},[\rho]_{\sim_C})=\llbracket \rho\rrbracket_\circ$
  follows as $\sim_C$ is a congruence compatible with
  the structure of the Brzozowski automaton, hence
  the map sending a state to its equivalence class preserves the
  accepted language.

  The fact that there are only finitely many reachable states from
  $[\rho]_{\sim_C}$ follows from the definition of $\sim_C$. More precisely,
  we defined it in terms of $\sim_B$ for which we already know that taking
  Brzozowski derivatives leads to finitely many reachable states. The fact
  that $\sim_C$ is given by equality on expressions underneath $(-)^\circ$
  means that they do not contribute anything.
\end{proof}

The following corollary is immediate.

\begin{cor}\label{cor:rationalIsRegular}
  Every rational lasso language is regular.
\end{cor}

\section{Regular Lasso Languages are Rational}\label{sec:5}

The previous section gives us one direction of Kleene's theorem for
lasso languages. We obtain the other direction by slightly modifying a result
by Calbrix \etal, which is interesting in its own right
as it gives insights into the connection between $\Omega$-automata
and nondeterministic B\"{u}chi automata \cite{calbrix:1994:ultimatelyPeriodicWords}.

Calbrix \etal\ show how to build the rational $\omega$-language accepted by an
$\Omega$-automaton $\mathcal{A}$ from some rational languages defined
on the basis of $\mathcal{A}$ \cite{calbrix:1994:ultimatelyPeriodicWords}.
This shows that every regular $\omega$-language is rational.

We first recall the results (Lemma
\ref{lem:UVOmegaWords} and Theorem
\ref{thm:omegaAutToOmegaExpression}) by
\cite{calbrix:1994:ultimatelyPeriodicWords} and then continue with
a similar result for lasso languages which shows that every regular
lasso language is rational. We begin with a technical lemma.

\begin{lem}[\protect{\cite[Lemma~5]{calbrix:1994:ultimatelyPeriodicWords}}]\label{lem:UVOmegaWords}
  Let $U,V$ be regular languages such that $\epsilon\not\in V$,
  $UV^\ast = U$ and $V^+=V$. Then
  for all $uv^\omega\in UV^\omega$, there exist $u'\in U$ and $v'\in V$ such that
  $uv^\omega = u'v'^\omega$.
\end{lem}

This fact is a simple application of the pigeonhole principle. Given an ultimately
periodic word $uv^\omega\in UV^\omega$, we can find  $u'\in U$ and
$v_1,\ldots,v_n,v_{n+1},\ldots,v_{n+m}\in V$ such that
$uv^\omega = u'v_1\ldots v_n(v_{n+1}\ldots v_{n+m})^\omega$ using the
pigeonhole principle. As $UV^\ast = U$ and $V^+=V$ it follows that
$u'v_1\ldots v_n\in U$ and $v_{n+1}\ldots v_{n+m}\in V$. Hence the lemma states
that for suitable $U$ and $V$, if $uv^\omega$ is an arbitrary ultimately
periodic word in $UV^\omega$, we may instead use an identical
ultimately periodic word $u'v'^\omega\in UV^\omega$ which on top also satisfies the
property that $u'\in U$ and $v'\in V$. Calbrix \etal\ use this lemma
to show Theorem \ref{thm:omegaAutToOmegaExpression} and we use it later
in Section \ref{sec:6} to prove Proposition \ref{prop:weakCoverMapYieldsWeakCovering}.

\begin{thm}[\cite{calbrix:1994:ultimatelyPeriodicWords}]\label{thm:omegaAutToOmegaExpression}
  Let $\mathcal{A}=(X,Y,\overline{x},\delta_1,\delta_2,\delta_3,F)$ be a finite $\Omega$-automaton.
  For $x\in X$ and $y\in F$ we abbreviate
  \begin{align*}
    S_x &= \{u\in\Sigma^\ast\mid \delta_1(\overline{x},u)=x\}, \\
    R_{x,y} &= \{u\in\Sigma^+\mid \delta_1(x,u)=x \text{ and } (\delta_2:\delta_3)(x,u)=\delta_3(y,u)=y\}. 
  \end{align*}
  Then
  \[
    L_\omega(\mathcal{A}) = \bigcup_{x\in X}\bigcup_{y\in F} S_x\cdot R_{x,y}^\omega.
  \] 
\end{thm}
The rational language $S_x$ consists of all words that lead from the
initial state $\overline{x}$ to $x$. As $R_{x,y}$ contains those words
which at the same time bring us from $x$ back to $x$ via $\delta_1$,
from $x$ to $y$ via $(\delta_2\col \delta_3)$ and from $y$ back to
itself via $\delta_3$, concatenating infinitely many such words traces
paths which intersect a final state infinitely often. Hence
a saturated lasso automaton can be seen as a nondeterministic B\"{u}chi
automaton \cite[Remark~21]{ciancia:2019:omegaAutomata}.

\begin{cor}\label{cor:omegaRegExpressionFromOmegaAutomata}
  Every regular $\omega$-language is rational: for every finite $\Omega$-automaton $\mathcal{A}$,
  one can construct a rational $\omega$-expression $T$ such
  that $L_\omega(\mathcal{A})=\llbracket T\rrbracket_\omega$.
\end{cor}

\begin{proof}
  The languages $S_x$ and $R_{x,y}$ constructed in Theorem
  \ref{thm:omegaAutToOmegaExpression} are regular and hence also
  rational. Therefore, we can find matching
  rational expressions which allow us to define a suitable rational
  $\omega$-expression.
\end{proof}

In order to obtain a similar result for lasso languages, we slightly modify the
definition of $R_{x,y}$ from Theorem \ref{thm:omegaAutToOmegaExpression}.
This shows that every regular lasso language is rational.

\begin{prop}\label{prop:regLassoFromLassoAutomaton}
  Let $\mathcal{A}=(X,Y,\overline{x},\delta_1,\delta_2,\delta_3,F)$ be a finite
  lasso automaton. For $x\in X$ and $y\in F$, let $S_x$ be defined as in
  Theorem \ref{thm:omegaAutToOmegaExpression} and let
  \[
    R_{x,y}=\{u\in\Sigma^+\mid (\delta_2\col\delta_3)(x,u)=y\}.
  \] Then
  \[
    L_\circ(\mathcal{A}) = \bigcup_{x\in X}\bigcup_{y\in F} S_x\cdot R_{x,y}^\circ.
  \] 
\end{prop}

\begin{proof}
  If $(u,v)\in L_\circ(\mathcal{A})$, then $(\delta_2\col\delta_3)(\delta_1(\overline{x},u),v)\in F$.
  If we choose $x=\delta_1(\overline{x},u)$ and $y=(\delta_2\col\delta_3)(x,v)\in F$,
  then $(u,v)\in S_x\cdot R_{x,y}^\circ$.

  For the other inclusion, let $(u,v)\in S_x\cdot R_{x,y}^\circ$ for some
  $x\in X$ and $y\in F$. By definition, we have $x=\delta_1(\overline{x},u)$
  and $y=(\delta_2\col\delta_3)(x,v)$. Hence
  $(\delta_2\col\delta_3)(\delta_1(\overline{x},u),v)\in F$ and
  $(u,v)\in L_\circ(\mathcal{A})$.
\end{proof}

\begin{eg}
  In this example we extract a rational lasso expression from the lasso
  automaton in Example \ref{eg:BrzozowskiConstruction} reproduced below:
  
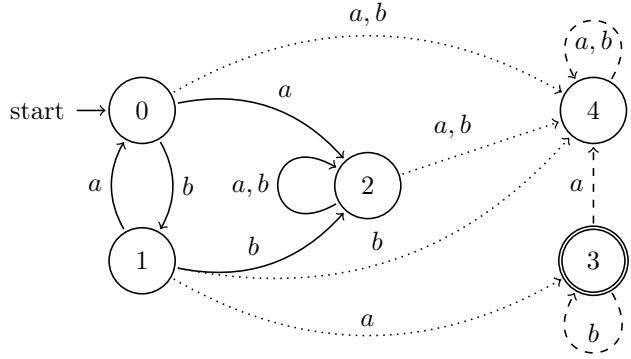
\begin{center}
    \begin{tikzpicture}[modal]
      \node[state,initial] at (0,2) (a) [label=above:{}] {$0$};
      \node[state] at (0,0) (b) [label=above:{}] {$1$};
      \node[state] at (3,1) (c) [label=above:{}] {$2$};
      \node[state] at (6,2) (d) [label=above:{}] {$4$};
      \node[state,accepting] at (6,0) (e) [label=above:{}] {$3$};
      \path[->] (a) edge[bend left] node[black]{$b$} (b);
      \path[->] (a) edge[bend left] node[black]{$a$} (c);
      \path[->,dotted] (a) edge[bend left] node[black]{$a,b$} (d);
      \path[->] (b) edge[bend left] node[black]{$a$} (a);
      \path[->] (b) edge[bend right] node[black]{$b$} (c);
      \path[->,dotted] (b) edge[bend right] node[black]{$b$} (d);
      \path[->,dotted] (b) edge[bend right] node[black]{$a$} (e);
      \path[->] (c) edge[reflexive left] node[left,black]{$a,b$} (c);
      \path[->,dotted] (c) edge node[black]{$a,b$} (d);
      \path[->,dashed] (d) edge[reflexive above] node[below,black]{$a,b$} (d);
      \path[->,dashed] (e) edge node[black]{$a$} (d);
      \path[->,dashed] (e) edge[reflexive below] node[above,black]{$b$} (e);
    \end{tikzpicture}
  \end{center}
  Instead of using rational languages, we immediately use corresponding
  rational expressions for the sake of simplicity.
  For each spoke state in  $\{0,1,2\} $ we compute $S_x$, giving
  \[
  S_0 = (ba)^\ast, \qquad S_1 = b(ab)^\ast, \qquad S_2 = (ba)^\ast a(a+b)^\ast + b(ab)^\ast b(a+b)^\ast.
  \] 
  Next, we compute $R_{x,y}$, which results in
  \[
  R_{0,3}=0, \qquad R_{1,3}=ab^\ast, \qquad R_{2,3}=0.
  \] 
  So the language accepted by the lasso automaton (which we call $\mathcal{A}$) is:
  \[
    L_\circ(\mathcal{A})=(ba)^\ast 0^\circ + b(ab)^\ast (ab^\ast)^\circ + \left[(ba)^\ast a(a+b)^\ast + b(ab)^\ast b(a+b)^\ast\right] 0^\circ.
  \] Note that this is provably equivalent to $b(ab)^\ast(ab^\ast)^\circ$ 
  which is exactly what we started out with.
\end{eg}

\begin{cor}\label{cor:ratLassoFromLassoAutomaton}
  Every regular lasso language is rational:
  for each finite lasso automaton $\mathcal{A}$, one can
  construct a rational lasso expression $\rho$ such that
  $L_\circ(\mathcal{A})=\llbracket\rho\rrbracket_\circ$.
\end{cor}

\begin{proof}
  The proof is analogous to that of Corollary
  \ref{cor:omegaRegExpressionFromOmegaAutomata}.
\end{proof}

\begin{thm}\label{thm:KleeneTheoremLassoLanguages}
  A lasso language is regular if and only if it is rational.
\end{thm}

\begin{proof}
  This theorem follows immediately from Corollary
  \ref{cor:rationalIsRegular} and Corollary
  \ref{cor:ratLassoFromLassoAutomaton}.
\end{proof}

\section{Rational $\omega$- and Lasso Expressions}\label{sec:6}

Additionally to the concepts already introduced in the preliminaries, we make
use of the following definitions and terminology throughout this section. As $\sim_\gamma$ is an
equivalence relation, we obtain a canonical map
$\phi:\Sigma^{\ast+}\to\Sigma^{\ast+}/{\sim_\gamma}$ which sends a lasso
$(u,v)$ to its $\gamma$-equivalence class $[(u,v)]_\gamma$. A subset
$L\subseteq \Sigma^{\ast+}$ is called \emph{$\sim_\gamma$-saturated} (or simply
\emph{saturated}) if $L=\phi^{-1}(K)$ for some $K\subseteq \Sigma^{\ast+}/{\sim_\gamma}$.
From the fact that $(u,v)\sim_\gamma (u',v')\iff uv^\omega = u'v'^\omega$ it
follows that $\Sigma^{\ast+}/{\sim_\gamma}$ is isomorphic to $\Sigma^{\text{up}}$
through the map $[(u,v)]_\gamma \mapsto uv^\omega$. We therefore
make no distinction between $\Sigma^{\ast+}/{\sim_\gamma}$ and
$\Sigma^{\text{up}}$, nor between $[(u,v)]_\gamma$ and $uv^\omega$.

In the first part of the section, we present some basic results on
saturation and show that a lasso language is saturated if and only if it is closed
under $\gamma$-expansion and $\gamma$-reduction.
Of particular interest are those rational lasso languages
whose semantics is saturated, as the image of such a language under
$\phi$ corresponds precisely to the ultimately periodic fragment of a rational
$\omega$-language. We lift this relationship between languages to that of
expressions by saying that a rational lasso expression $\tau$ represents a
rational $\omega$-expression $T$, if $\llbracket\tau\rrbracket_\circ$ is
saturated and its image under $\phi$ is the ultimately periodic fragment
of $\llbracket T\rrbracket_\omega$. The question we then try to answer is,
given $T$, can we syntactically build $\tau$? We show that this is indeed
possible under the assumption that we have two additional operations on rational
expressions. Our approach consists of first introducing a weaker notion, that of
weak representability, which drops the saturation requirement. From $T$ we
then construct in a first step a weak representation $\tau$ whose semantics
is not saturated but is only closed under $\gamma$-expansion. This first
construction makes use of the sequential splitting relation, which to each
rational expression $t$ associates a finite set of pairs of rational expressions
called splits, intuitively corresponding to all the ways words in the semantics
of $t$ can be split into two.
We subsequently modify our weak representation $\tau$ to obtain a rational lasso expression
which is also closed under $\gamma$-reduction and hence saturated. This part
of the construction makes use of two additional operations on
rational expressions, intersection and root, which both preserve rationality.

The use of the sequential splitting relation and the rational operations
intersection and root has to do with the fact that we have to capture the
rewrite rules $\gamma_1$ and $\gamma_2$ on the level of expressions, that is
syntactically.

The next lemma introduces several equivalent definitions of saturation using
the canonical map $\phi$ which we described above. The most interesting for
our purposes is the last one, which describes saturation in terms of the
rewrite rules $\to_{\gamma_1}$ and $\to_{\gamma_2}$.

\begin{lem}\label{lem:saturationLemma}
  Let $L\subseteq \Sigma^{\ast +}$. The following are equivalent definitions
  of $\sim_\gamma$-saturation:
\begin{enumerate}
  \item $L=\phi^{-1}(K)$ for some $K\subseteq \Sigma^{\ast +}/{\sim_\gamma}$.
  \item $L$ is the union of $\sim_\gamma$-equivalence classes.
  \item $L=\phi^{-1}(\phi(L))$.
  \item Let
    \begin{align*}
      L^\uparrow &=\{(u',v')\in\Sigma^{\ast+}\mid \exists i\in \{1,2\}, \exists (u,v)\in L: (u',v')\to_{\gamma_i}(u,v) \},\\
      L^\downarrow &= \{(u',v')\in\Sigma^{\ast+}\mid\exists i\in \{1,2\}, \exists (u,v)\in L: (u,v)\to_{\gamma_i}(u',v')\}.
    \end{align*}
    Then $L$ is saturated if and only if $L=L^\uparrow=L^\downarrow$.
\end{enumerate}
\end{lem}

\begin{proof}
Of these four, the first three definitions are standard. The last
definition of saturation states that if $L$ is closed under
$\gamma$-expansion ($L=L^\uparrow$) and under $\gamma$-reduction
$L=L^\downarrow$, then $L$ is saturated. To see why this is the case
assume that $L=L^\uparrow=L^\downarrow$, let $(u,v)\in L$ and
$(u',v')\sim_\gamma (u,v)$. If we show that $(u',v')\in L$ it follows immediately
that $L$ is a union of equivalence classes. From $(u,v)\sim_\gamma (u',v')$ 
it follows that the two lassos reduce to the same normal form $(u'',v'')$.
As $L$ is closed under $\gamma$-reduction and $(u,v)\to_{\gamma}(u'',v'')$,
$(u'',v'')\in L$. Next as $L$ is also closed under $\gamma$-expansion and
$(u',v')\to_\gamma (u'',v'')$, $(u',v')\in L$ as required. Hence
$L$ is saturated.
\end{proof}

We remind again that we make no distinction between $\Sigma^{\ast+}/{\sim_\gamma}$ 
and $\Sigma^{\text{up}}$ as they are isomorphic so that we are treating
$\phi$ as though its type was $\phi:\Sigma^{\ast+}\to \Sigma^{\text{up}}$.

\begin{defn}\label{defn:covering}
  For a rational lasso expression $\tau\in\text{Exp}_\circ$ and a rational
  $\omega$-expression $T\in \text{Exp}_\omega$, 
  \begin{enumerate}
    \item $\tau$ \emph{weakly represents} $T$ if $\phi\left(
        \llbracket\tau\rrbracket_\circ \right)=\text{UP}(\llbracket
        T\rrbracket_\omega) $,
    \item $\tau$ \emph{represents} $T$ if $\llbracket\tau\rrbracket_\circ = \phi^{-1}\left(\text{UP}(\llbracket T\rrbracket_\omega)\right)$.
  \end{enumerate}
\end{defn}

\begin{eg}
  Let $T=(a+b)^\ast a^\omega$. Then
  $((a+b)^\ast,a)$ constitutes a weak representation but not a representation
  of $T$. This is seen from the following facts:
  \begin{enumerate}
    \item $\phi\left(\llbracket ((a+b)^\ast,a)\rrbracket_\circ\right) = \phi(\{(u,a)\mid u\in\Sigma^\ast\}) = \{ua^\omega\mid u\in\Sigma^\ast\}=\text{UP}(\llbracket T\rrbracket_\omega)$,
    \item $\phi(\epsilon,aa)=a^\omega\in \text{UP}(\llbracket T\rrbracket_\omega)$ but $(\epsilon,aa)\not\in\llbracket ((a+b)^\ast,a)\rrbracket_\circ$.
  \end{enumerate}
\end{eg}

\begin{rmk}\label{rmk:rmk1}
  If $\tau$ represents $T$ it also weakly represents $T$ as $\phi$ is
  surjective, i.e. $\phi\circ \phi^{-1}=\text{id}$. Conversely, if $\tau$ weakly represents $T$, then it
  represents $T$ if and only if
  $\llbracket \tau\rrbracket_\circ =
  \phi^{-1}(\phi(\llbracket\tau\rrbracket_\circ))$, that is if and only if
  $\llbracket\tau\rrbracket_\circ$ is
  $\sim_\gamma$-saturated. In the
  previous example, $((a+b)^\ast,a^+)$ is a weak representation of $T$
  and moreover its semantics is $\sim_\gamma$-saturated, hence it is a
  representation of $T$.
\end{rmk}

As discussed at the start of the section, we proceed by first constructing
a weak representation whose semantics is closed under $\gamma$-expansion. This means
we want to syntactically mimic the following directions of rules
$\gamma_1$ and $\gamma_2$ :
\begin{align*}
  (u,vw) &\to (uv,wv) & (u,v) &\to (u,v^k)\quad (k\geq 1)
\end{align*}
For the first expansion, given a lasso $(u,v)$ we may split $v$ into two words
$v_1$ and $v_2$ and obtain the new lasso $(uv_1,v_2v_1)$. For the second one,
given a lasso $(u,v)$ we may choose arbitrary $k\geq 1$ and obtain
the new lasso $(u,v^k)$. To mimic the first expansion, we need to be able to split
rational expressions into two rational expressions. Then given a rational
lasso expression $(t,r)$ (recall that we use $(t,r)$ to mean $t\cdot r^\circ$),
we can split $r$ into two rational expressions $r_1$ and $r_2$ and obtain
a new rational lasso expression $(tr_1,r_2r_1)$. For the second expansion,
we can use the Kleene plus to go from $(t,r)$ to $(t,r^+)$, taking care
of the arbitrary $k\geq 1$.

Hence, in order to close the semantics under $\gamma$-expansion, whenever
we add a term $(t,r)$, we also want to add $(tr_1,r_2r_1)$ for all possible
splits of $r$, and additionally, we also want to apply the Kleene plus.
We now formally define the sequential splitting relation and show some of
its properties which are needed to show that our first construction yields
a weak representation.

\begin{defn}[\cite{kappe:2017:concurrentKleeneAlgebra}]\label{defn:sequentialSplittingRelation}
  Let $\nabla:\text{Exp} \to 2^{\text{Exp}\times \text{Exp}}$ be defined
  recursively:
	\begin{align*}
		\nabla(0) &= \emptyset & \nabla(1) &= \{(1,1)\} & \nabla(a) &= \{(1,a),(a,1)\}
	\end{align*}
	\[
		\nabla(t+r)=\nabla(t)\cup \nabla(r)\qquad\qquad\qquad \nabla(t^\ast) = \{(t^\ast\cdot t_0,t_1\cdot t^\ast)\mid (t_0,t_1)\in\nabla(t)\}\cup \{(1,1),(t^\ast \cdot t,1)\}  
	\]\phantom{test}
	\vspace{-0.2cm}
	\[
	\nabla(t\cdot r) = \{(t_0,t_1\cdot r)\mid (t_0,t_1)\in \nabla(t)\}\cup \{(t\cdot r_0,r_1)\mid (r_0,r_1)\in\nabla(r)\} 
	\] 
	We write $\nabla_t$ for $\nabla(t)$ and call $\nabla_t$ the \emph{sequential splitting relation of $t$}.
\end{defn}

The following lemma establishes some properties of the splitting relation,
of which points $1.$ and $2.$ are taken from \cite{kappe:2017:concurrentKleeneAlgebra}.

\begin{lem}\label{lem:sequentialSplittingBasicProperties}
	Let $t\in \text{Exp}$. The sequential splitting relation $\nabla_t$ satisfies:
	\begin{enumerate}[label={\arabic*.},ref={6.6.\arabic*}]
		\item $|\nabla_t|$ is finite and $\forall (t_0,t_1)\in\nabla_t\colon t_0\cdot t_1 \leq t$,
    \item if $u\cdot v\in\llbracket t\rrbracket$, then there is a split $(t_0,t_1)\in\nabla_t$
			such that $u\in \llbracket t_0\rrbracket$ and $v\in\llbracket t_1\rrbracket$,
    \item if $(t_0,t_1)\in\nabla_t$ and $(r_0,r_1)\in\nabla_{t_1}$, then there exists
      $(s_0,s_1)\in\nabla_t$ such that $t_0\cdot r_0\leq s_0$ and $r_1\leq s_1$. \label{lem:sequentialSplittingTransfer}
	\end{enumerate}
\end{lem}

\begin{proof}
  The first two items follow by
  \cite{kappe:2017:concurrentKleeneAlgebra},
  so we only focus on the third
  item, which we show by structural induction on rational
  expressions. For $t=0$, $\nabla_0=\emptyset$ and so it trivially
  holds. For $t=1$, $\nabla_1=\{(1,1)\}$ and so there is only one
  choice for $(t_0,t_1)$ and $(r_0,r_1)$, namely the pair $(1,1)$. As
  $1\cdot 1\leq 1$ and $1\leq 1$, we are done. For $t=a$, there are
  three cases. First, let $(t_0,t_1)=(a,1)$, then $(r_0,r_1)=(1,1)$
  and we have $a\cdot 1\leq a$ and $1\leq 1$. For the other two cases,
  let $(t_0,t_1)=(1,a)$. If $(r_0,r_1)=(1,a)$, then we are done as
  $1\cdot 1\leq 1$ and $a\leq a$. On the other hand, for
  $(r_0,r_1)=(a,1)$ we have $1\cdot a\leq a$ and $1\leq 1$. This
  covers all the base cases.

	For the induction steps, we start with $t\cdot r$. There are two cases two consider, the first
	split can either be $(t_0,t_1\cdot r)$ or $(t\cdot r_0,r_1)$. For the case $(t_0,t_1\cdot r)$,
	there are again two subcases, namely the splits $(t_{10},t_{11}\cdot r)$ and
	$(t_1\cdot r_0,r_1)$ in $\nabla_{t_1\cdot r}$. For the first subcase, we know that
	$(t_{10},t_{11})\in\nabla_{t_1}$, so by the induction hypothesis, there exists a split
	$(s_0,s_1)\in\nabla_t$ such that $t_0\cdot t_{10}\leq s_0$ and $t_{11}\leq s_1$. It follows
	that $(s_0,s_1\cdot r)\in\nabla_{t\cdot r}$ is a suitable candidate as $t_0\cdot t_{10}\leq s_0$ and
	$t_{11}\cdot r\leq s_1\cdot r$. For the second subcase, we have $(r_0,r_1)\in\nabla_r$, hence
	$(t\cdot r_0,r_1)\in\nabla_{t\cdot r}$ works as $t_0\cdot t_1\cdot r_0\leq t\cdot r_0$ 
	and $r_1\leq r_1$. Finally, for the case $(t\cdot r_0,r_1)$, let $(r_{10},r_{11})\in\nabla_{r_1}$.
	Then, by the induction hypothesis, there exists a split $(s_0,s_1)\in\nabla_r$ with
	$r_0\cdot r_{10}\leq s_0$ and $r_{11}\leq s_1$. Our candidate is $(t\cdot s_0,s_1)$ which
	works as $t\cdot r_0\cdot r_{10}\leq t\cdot s_0$ and $r_{11}\leq s_1$.

	The next induction step is $t_0+t_1$. Let $(t,t')\in\nabla_{t_0+t_1}$, then
	$(t,t')\in\nabla_{t_i}$ for some $i\in \{0,1\}$. Next, take $(r_0,r_1)\in\nabla_{t'}$.
	By the induction hypothesis we find $(s_0,s_1)\in\nabla_{t_i}$ such that
	$t\cdot r_0\leq s_0$ and $r_1\leq s_1$. As $(s_0,s_1)\in\nabla_{t_i}$, we also
	have $(s_0,s_1)\in\nabla_{t_0+t_1}$. Choosing $(s_0,s_1)$ suffices to satisfy the claim.

  For the last case, we consider the expression $t^\ast$. If we consider
  $(1,1)\in\nabla_{t^\ast}$, then this is taken care of by one of the previous
  points. A similar argument also suffices for $(t^\ast\cdot t,1)\in\nabla_{t^\ast}$.
  For the other case, let
  $(t^\ast\cdot t_0,t_1\cdot t^\ast)\in\nabla_{t^\ast}$ where $(t_0,t_1)\in\nabla_t$.
	There are four cases, as the splits in $\nabla_{t_1\cdot t^\ast}$ are either of the
	form $(t_1\cdot 1,1)$, $(t_1\cdot t^\ast\cdot t,1)$, $(t_1\cdot t^\ast\cdot r_0,r_1\cdot t^\ast)$
	(with $(r_0,r_1)\in\nabla_{t}$) or $(t_{10},t_{11}\cdot t^\ast)$
	(with $(t_{10},t_{11})\in \nabla_{t_1}$). For the first case, consider
  $(t^\ast\cdot t,1)\in\nabla_{t^\ast}$. Then $t^\ast\cdot t_0\cdot t_1\cdot 1\leq t^\ast\cdot t$ 
  and $1\leq 1$ as required. For the second case consider again $(t^\ast\cdot t,1)\in\nabla_{t^\ast}$ 
  for which we have $t^\ast\cdot t_0\cdot t_1\cdot t^\ast\cdot t\leq t^\ast\cdot t$ 
  and $1\leq 1$ satisfying the requirement. For the third case we choose
  the split $(t^\ast\cdot r_0,r_1\cdot t^\ast)\in\nabla_{t^\ast}$ as
  $t^\ast\cdot t_0\cdot t_1\cdot t^\ast\cdot r_0\leq t^\ast\cdot r_0$ and
  $r_1\cdot t^\ast\leq r_1\cdot t^\ast$. Finally, for the last case we find
  by the induction hypothesis a split $(r_0,r_1)\in\nabla_t$ with
  $t_0\cdot t_{10}\leq r_0$ and $t_{11}\leq r_1$ for which we then have
  $t^\ast\cdot t_0\cdot t_{10}\leq t^\ast\cdot r_0$ and
  $t_{11}\cdot t^\ast\leq r_1\cdot t^\ast$ with $(t^\ast\cdot r_0,r_1\cdot t^\ast)\in\nabla_{t^\ast}$ 
  as desired.
\end{proof}

\begin{eg}
  For this example we explore the sequential splitting relation of a simple
  rational expression and briefly touch on the different properties
  outlined in Lemma \ref{lem:sequentialSplittingBasicProperties}.
  Let $t=b(a+b^\ast)$. Then
   \[
  \nabla_t = \{(1,b(a+b^\ast)),(b,1(a+b^\ast)),(b1,a),(ba,1),(bb^\ast1,bb^\ast),(bb^\ast b,1b^\ast),(b1,1),(bb^\ast b,1)\}.
  \] 
  If we choose a word in $\llbracket t\rrbracket$, say $bb\cdot b$, then we can
  find a corresponding split, in this case $(bb^\ast b,1b^\ast)$ for example.
  Moreover,  if we choose any split in $\nabla_t$, say $(bb^\ast b,1b^\ast)$,
  and take any split in $\nabla_{1b^\ast}$, say $(1b^\ast 1,bb^\ast)$, then
  we can find a split  $(t_0,t_1)\in\nabla_t$ such that
  $bb^\ast b 1b^\ast 1\leq t_0$ and $b b^\ast \leq t_1$, namely $(bb^\ast 1,bb^\ast)$.
\end{eg}

With the sequential splitting relation at hand, we define the map
$h$ which takes a rational $\omega$-expression and returns a rational
lasso expression which is weakly representing. The definition of $h$ makes
use of the sequential splitting relation and the Kleene plus in the definition
of $h(t^\omega)$. Instead of just mapping
$t^\omega$ to $(1,t)$ (i.e.\ $t^\circ$), we go over every possible way of splitting $t$ into
$t_0$ and $t_1$ and make sure to add $(t_0,t_1t_0)$. Additionally, we also
incorporate the Kleene plus, which appears as a Kleene star in the definition.

\begin{defn}\label{defn:weakCoverMap}
	Let $h:\text{Exp}_\omega\to \text{Exp}_\circ$ be defined as
	\begin{align*}
		h(0)&=0 & h(t^\omega)&= \sum_{(t_0,t_1)\in \nabla_t} (t^\ast\cdot t_0,t_1\cdot t^\ast\cdot t_0) \\ 
		h(T_1+T_2)&=h(T_1)+h(T_2) & h(t\cdot T)&= t\cdot h(T)
	\end{align*}
\end{defn}

 Rational lasso expressions can only contain finite sums, so it is
 crucial in the definition of $h(t^\omega)$ that $|\nabla_t|$ is
 finite (Lemma
 \ref{lem:sequentialSplittingBasicProperties}). The next steps are to show
 that $h(T)$ weakly represents $T$ and that its semantics is closed under
 $\gamma$-expansion. The proof that $h(T)$ weakly represents $T$ relies on
 Lemma \ref{lem:UVOmegaWords} from Section \ref{sec:5}.

\begin{prop}\label{prop:weakCoverMapYieldsWeakCovering}
	Let $T\in\text{Exp}_\omega$. Then $h(T)$ weakly represents $T$.
\end{prop}

\begin{proof}
  This is shown by structural induction on rational $\omega$-expressions.
  The case $T=0$ is trivial. For $T=t^\omega$, let first
  $uv^\omega\in\text{UP}(\llbracket T\rrbracket_\omega)$. Then 
  \[
    uv^\omega\in \llbracket t\rrbracket^\ast (\llbracket t\rrbracket^+)^\omega
  \] and by Lemma \ref{lem:UVOmegaWords} there exist $u',v'$ with
  $uv^\omega=u'v'^\omega$, $u'\in\llbracket t\rrbracket^\ast$ and
  $v'\in\llbracket t\rrbracket^+$. As $v'\in \llbracket t\rrbracket^+$
  we can find $v_1\in\llbracket t\rrbracket$ and $w\in\llbracket t\rrbracket^\ast$ 
  with $v'=v_1w$. Now $v_1=\epsilon\cdot v_1$ gives rise to a split
  $(t_0,t_1)\in\nabla_t$ with $\epsilon\in\llbracket t_0\rrbracket$ and
  $v_1\in\llbracket t_1\rrbracket$. It is quickly verified that
  \[
    (u',v') = (u'\cdot \epsilon,v_1\cdot w\cdot \epsilon) \in \llbracket (t^\ast t_0,t_1t^\ast t_0)\rrbracket_\circ \subseteq \llbracket h(t^\omega)\rrbracket_\circ.
  \] Hence
  $uv^\omega=u'v'^\omega\in \phi\left(\llbracket h(t^\omega)\rrbracket_\circ\right)$.
  For the other direction, let
  $(u,v)\in\llbracket h(t^\omega)\rrbracket_\circ$. Then
  there exists a split $(t_0,t_1)\in\nabla_t$ with
  $u\in \llbracket t^\ast t_0\rrbracket$ and
  $v\in \llbracket t_1t^\ast t_0\rrbracket$. Now
  \[
    \llbracket t^\ast t_0\rrbracket\cdot (\llbracket t_1 t^\ast t_0\rrbracket)^\omega = 
    \llbracket t^\ast t_0 t_1\rrbracket\cdot (\llbracket t^\ast t_0 t_1\rrbracket)^\omega \subseteq \llbracket t^\omega\rrbracket_\omega,
  \] hence $\phi(u,v)=uv^\omega\in\llbracket t^\omega\rrbracket_\omega$.

  For the induction steps, we first consider the case $ T_1+ T_2$. For
  the right to left inclusion, let
  $uv^\omega\in\llbracket T_1+ T_2\rrbracket_\omega$. Then
  $uv^\omega\in\llbracket T_i\rrbracket_\omega$ for
  $i\in \{1,2\}$. By the induction hypothesis, there exists some
  $(u',v')\in\llbracket h(T_i)\rrbracket_\circ$ with
  $(u,v)\sim_\gamma (u',v')$. Then
	\[
		(u',v')\in\llbracket h( T_1)\rrbracket\cup\llbracket h( T_2)\rrbracket_\circ = \llbracket h( T_1)+h( T_2)\rrbracket_\circ = \llbracket h( T_1+ T_2)\rrbracket_\circ
	\] as required. For the other inclusion, let $(u,v)\in\llbracket h( T_1+ T_2)\rrbracket_\circ$, then
	$(u,v)\in\llbracket h( T_i)\rrbracket_\circ$ for some $i\in \{1,2\}$. By the induction
	hypothesis, $uv^\omega\in\llbracket T_i\rrbracket_\omega \subseteq \llbracket T_1+ T_2\rrbracket_\omega$.

	The last case is that of $t\cdot T$. For the right to left inclusion, take
	$u_0u_1v^\omega\in\llbracket t\cdot  T\rrbracket_\omega$ where
	$u_0\in\llbracket t\rrbracket$ and $u_1v^\omega\in\llbracket T\rrbracket_\omega$. By the
	induction hypothesis, there exists $(u',v')\in\llbracket h( T)\rrbracket_\circ$ such that
	$uv^\omega = u'v'^\omega$. It follows that
	\[
		(u_0u',v')\in\llbracket t\cdot h( T)\rrbracket_\circ = \llbracket h(t\cdot  T)\rrbracket_\circ
	\] and moreover, $u_0u'v'^\omega=u_0u_1v^\omega$ and we are done. For the other
  inclusion, let
	$(u_0u_1,v)\in\llbracket h(t\cdot  T)\rrbracket_\circ$, where $u_0\in\llbracket t\rrbracket$ 
	and $(u_1,v)\in\llbracket h( T)\rrbracket_\circ$. By the induction hypothesis,
	$u_1v^\omega\in\llbracket  T\rrbracket_\omega$ and so
	$u_0u_1v^\omega\in\llbracket t\cdot T\rrbracket_\omega$ which concludes the proof.
\end{proof}

Showing that the semantics of $h(T)$ is closed under $\gamma$-expansion
relies crucially on some of the properties of the sequential splitting
relation, in particular on Lemma \ref{lem:sequentialSplittingTransfer}.

\begin{prop}\label{prop:weakCoverMapExpansionClosed}
	Let $T\in\text{Exp}_\omega$. Then $\llbracket h(T)\rrbracket_\circ$ is closed
	under $\gamma$-expansion.
      \end{prop}

\begin{proof}
	We proceed by structural induction on rational $\omega$-expressions. The first base case
	$T=0$ is trivial as $\llbracket h(0)\rrbracket_\circ=\emptyset$. For the
	second, $ T=t^\omega$. We show that $\llbracket h(t^\omega)\rrbracket_\circ$ 
	is closed under $\gamma_1$- and $\gamma_2$-expansion. The result then follows as these
	commute. For $\gamma_1$-expansion, let $(ua,va)\to_{\gamma_1}(u,av)$ and
	$(u,av)\in\llbracket h(t^\omega)\rrbracket_\circ$. Then there exists a split
	$(t_0,t_1)\in\nabla_t$ such that $u\in\llbracket t^\ast\cdot t_0\rrbracket$ and
	$av\in\llbracket t_1\cdot t^\ast\cdot t_0\rrbracket$. Hence we can find
	$v_1,\ldots,v_k\in\Sigma^\ast$ with $v_1\in\llbracket t_1\rrbracket$,
	$v_2,\ldots,v_{k-1}\in \llbracket t^\ast\rrbracket$, $v_k\in\llbracket t_0\rrbracket$ and
	$v_1v_2\ldots v_k=av$.

	We need to distinguish three subcases:
	\begin{enumerate}
		\item If $v_1\not =\epsilon$, then $v_1=av_1'$ and this induces a split
			$(r_0,r_1)\in\nabla_{t_1}$, where $a\in\llbracket r_0\rrbracket$ and
			$v_1'\in\llbracket r_1\rrbracket$. By Lemma \ref{lem:sequentialSplittingTransfer}, there exists
			a split $(s_0,s_1)\in\nabla_t$ such that $t_0\cdot r_0\leq s_0$ and
			$r_1\leq s_1$. We now have that
      \[
				ua\in\llbracket t^\ast\cdot t_0\cdot r_0\rrbracket \subseteq \llbracket t^\ast\cdot s_0\rrbracket \qquad \text{ and }\qquad
				va\in\llbracket r_1\cdot t^\ast\cdot t_0\cdot r_0\rrbracket\subseteq \llbracket s_1\cdot t^\ast\cdot s_0\rrbracket.
      \]
			Hence $(ua,va)\in\llbracket h(t^\omega)\rrbracket_\circ$.
		\item If $v_1=\epsilon$ and $k=2$, we have that $v_2=av$, and hence a split
			$(r_0,r_1)\in\nabla_{t_0}$ where $a\in\llbracket r_0\rrbracket$ and
			$v\in\llbracket r_1\rrbracket$. Moreover, as $v_1=\epsilon$, it follows that
			$1\leq t_1$ and so $t_0 = t_0\cdot 1\leq t_0\cdot t_1\leq t$. Thus
			\[
				ua\in\llbracket t^\ast\cdot t_0\cdot r_0\rrbracket \subseteq \llbracket t^\ast\cdot e\cdot r_0\rrbracket \subseteq \llbracket t^\ast\cdot r_0\rrbracket\qquad \text{ and }\qquad
				va\in\llbracket r_1\cdot r_0\rrbracket\subseteq \llbracket r_1\cdot t^\ast\cdot r_0\rrbracket.
      \]
			Hence $(ua,va)\in\llbracket h(t^\omega)\rrbracket_\circ$.
		\item If $v_1=\epsilon$ and $k>2$, then $v_2=av_2'$ as $\epsilon\not\in\llbracket t\rrbracket$. This gives
			rise to a split $(r_0,r_1)\in\nabla_t$ with $a\in \llbracket r_0\rrbracket$ and
			$v_2'\in\llbracket r_1\rrbracket$. As in the previous case, we have $t_0\leq t$ as
			$v_1=\epsilon$. This gives us that
			\begin{align*}
				ua&\in\llbracket t^\ast\cdot t_0\cdot r_0\rrbracket \subseteq \llbracket t^\ast\cdot t\cdot r_0\rrbracket\subseteq \llbracket t^\ast\cdot r_0\rrbracket,\\
				va&\in\llbracket r_1\cdot t^\ast\cdot t_0\cdot r_0\rrbracket\subseteq \llbracket r_1\cdot t^\ast\cdot t\cdot r_0\rrbracket\subseteq \llbracket r_1\cdot t^\ast\cdot r_0\rrbracket.
			\end{align*}
			So $(ua,va)\in \llbracket h(t^\omega)\rrbracket_\circ$.
	\end{enumerate}
	For $\gamma_2$-expansion, let $(u,v^k)\to_{\gamma_2}(u,v)$ ($k\geq 1$) and
	$(u,v)\in\llbracket h(t^\omega)\rrbracket_\circ$. So there exists a split $(t_0,t_1)\in\nabla_t$
	with  $u\in\llbracket t^\ast\cdot t_0\rrbracket$ and
	$v\in\llbracket t_1\cdot t^\ast\cdot t_0\rrbracket$. It suffices to show that
	$v^k\in\llbracket t_1\cdot t^\ast\cdot t_0\rrbracket$. We have
	\[
	v^k\in\llbracket t_1\cdot t^\ast\cdot t_0\rrbracket^k = \llbracket (t_1\cdot t^\ast\cdot t_0)^k\rrbracket \subseteq \llbracket t_1\cdot t^\ast\cdot t_0\rrbracket,
	\] where the last step holds because
	\[
		(t_1\cdot t^\ast\cdot t_0)\cdot (t_1\cdot t^\ast\cdot t_0) = t_1\cdot t^\ast\cdot (t_0\cdot t_1)\cdot t^\ast\cdot t_0 \leq t_1\cdot t^\ast\cdot t\cdot t^\ast\cdot t_1\leq t_1\cdot t^\ast\cdot t_0,
	\] as $t_0\cdot t_1\leq t$. This concludes the base cases.

	For the first induction step, let $ T= T_1+ T_2$. Let $(u,v)\to_\gamma (u',v')$ and
	$(u',v')\in\llbracket h( T_1+ T_2)\rrbracket_\circ$. As
	$\llbracket h( T_1+ T_2)\rrbracket_\circ=\llbracket h( T_1)\rrbracket\cup\llbracket h( T_2)\rrbracket_\circ$, $(u',v')\in\llbracket h( T_i)\rrbracket_\circ$ for $i\in \{1,2\}$. By the induction
	hypothesis, $(u,v)\in\llbracket h( T_i)\rrbracket_\circ$ and so
	$(u,v)\in\llbracket h( T_1+ T_2)\rrbracket_\circ$.

	For the remaining induction step, let $ T=t\cdot T'$. We treat the case of
	$\gamma_1$- and $\gamma_2$-expansion separately. For $\gamma_1$-expansion,
	let $(ua,va)\to_{\gamma_1}(u,av)$ and $(u,av)\in\llbracket h(t\cdot T')\rrbracket_\circ$.
	As
	$\llbracket h(t\cdot T')\rrbracket_\circ=\llbracket t\rrbracket\cdot \llbracket h( T')\rrbracket_\circ$ 
	there exist $u_0,u_1\in\Sigma^\ast$ such that $u_0u_1=u$, $u_0\in\llbracket t\rrbracket$ and
	$(u_1,av)\in\llbracket h( T')\rrbracket_\circ$. By the induction hypothesis,
	$(u_1a,va)\in\llbracket h( T')\rrbracket_\circ$, and
	\[
		(ua,va)=(u_0u_1a,va)\in\llbracket t\rrbracket \cdot \llbracket h( T')\rrbracket_\circ = \llbracket h(t\cdot T')\rrbracket_\circ.
	\] Finally, for $\gamma_2$-expansion, let $(u,v^k)\to_{\gamma_2}(u,v)$ ($k\geq 1$) and
	$(u,v)\in\llbracket h(t\cdot T')\rrbracket_\circ$. So there exist $u_0,u_1$ with
	$u_0\in\llbracket t\rrbracket$, $(u_1,v)\in\llbracket h( T')\rrbracket_\circ$ and $u_0u_1=u$.
	By the induction hypothesis, $(u_1,v^k)\in\llbracket h( T')\rrbracket_\circ$ and so
	$(u,v^k)=(u_0u_1,v^k)\in\llbracket h(t\cdot T')\rrbracket_\circ$.
\end{proof}

Our next goal is to take the weak representation, of which we know that its
semantics is closed under $\gamma$-expansion, and from it construct a new
rational lasso expression, whose semantics is also closed under $\gamma$-reduction.
So now we want to mimic the other directions of $\gamma_1$ and $\gamma_2$,
namely:
\begin{align*}
  (uv,wv)&\to (u,vw) & (u,v^k)&\to (u,v)\quad (k\geq 1)
\end{align*}
For the first reduction, given a lasso $(u,v)$, if we split $u$ and
$v$ into words $u_1,u_2,v_1$ and $v_2$, and if $u_2$ happens to be equal
to $v_2$, we may reduce to the lasso  $(u_1,u_2v_1)$ (or equivalently  $(u_1,v_2v_1)$).
For the other reduction, if we have a lasso $(u,v)$ and it happens to be
the case that $v=w^k$ for some $k\geq 1$, then we may reduce to the new lasso
$(u,w)$. On the level of expressions, if we start with $(t,r)$ and split
both $t$ and $r$ into $t_0,t_1,r_0,r_1$ and it so happens that
$t_1\cap r_1\not =\emptyset$, then we may form the new rational lasso
expression $(t_0,(t_1\cap r_1)r_0)$. For the other reduction, we make use
of the root operation, defined on a rational language $U$ as
$\sqrt{U}=\{ u\in\Sigma^+\mid \exists k\geq 1:u^k\in U\}$. This then allows
us to go from $(t,r)$ to $(t,\sqrt{r})$, taking care of the reduction.
The next definition combines the ideas of using sequential splits, the
root operation and the intersection. We point out that we assume the input
to be a disjunctive form.

\begin{defn}
  Let $\tau=\sum_{i=1}^n t_i\cdot r_i^\circ\in \text{Exp}_\circ$. We
  define the map $\Gamma$ as
  \[
    \Gamma(\tau) = \sum_{i=1}^n \sum_{\substack{(t'_0,t'_1)\in\nabla_t \\ (s'_0,s'_1)\in\nabla_s}} t'_0\cdot \left( \sqrt{(t'_1\cap s'_1)\cdot s'_0} \right)^\circ.
  \] 
\end{defn}

We follow up with a proposition which relates the semantics of the rational
lasso expression $\tau$ to that of $\Gamma(\tau)$. From this we
can derive a couple of corollaries which prove useful later on in showing
that if the semantics of $\tau$ is closed under $\gamma$-expansion, then
that of $\Gamma(\tau)$ is saturated. The intuition behind the proposition is
gained by looking at the shape of the lassos on either side of the if and only if.
Specifically, for the lasso  $(u,v)$ on the left-hand side, any lasso which
is $\gamma$-equivalent to it has to be of the shape as shown on the right-hand side.

\begin{prop}\label{prop:connectionRhoGamma}
  Let
  $\tau=\Sigma_{i=1}^n t_i\cdot r_i^\circ\in
  \text{Exp}_\circ$. Then
  \[
    (u,v) \in \llbracket \Gamma(\tau)\rrbracket_\circ \iff \exists k_1,k_2 \geq 0, \exists v_1,v_2\in\Sigma^\ast: v=v_1v_2 \land (uv^{k_1}v_1,v_2v^{k_2+k_1}v_1)\in \llbracket \tau\rrbracket_\circ.
  \] 
\end{prop}

\begin{proof}
  This is shown by using the various definitions of $\Gamma$, the root operation,
  the intersection and the semantics of expressions.
  \begin{align*}
    (u,v)\in\llbracket\Gamma(\tau)\rrbracket_\circ &\iff \exists i,\exists (t'_0,t'_1)\in\nabla_{t_i},\exists (s'_0,s'_1)\in\nabla_{s_i}: (u,v)\in\left\llbracket t'_0\cdot \left( \sqrt{ (t'_1\cap s'_1)\cdot s'_0 } \right)^\circ \right\rrbracket_\circ \\
    &\iff \exists i,\exists (t'_0,t'_1)\in\nabla_{t_i},\exists (s'_0,s'_1)\in\nabla_{s_i},\exists k\geq 1: u\in \llbracket t'_0\rrbracket, v^k\in\llbracket (t'_1\cap s'_1)\cdot s'_0 \rrbracket \\
    &\iff \exists i,\exists (t'_0,t'_1)\in\nabla_{t_i},\exists (s'_0,s'_1)\in\nabla_{s_i},\exists k_1,k_2\geq 0,\exists v_1,v_2\in\Sigma^\ast:\\
    &\qquad\qquad v_1v_2=v, u\in \llbracket t'_0\rrbracket, v^{k_1}v_1\in\llbracket t'_1\cap s'_1 \rrbracket, v_2v^{k_2}\in\llbracket s'_0\rrbracket \\
    &\iff \exists i,\exists k_1,k_2\geq 0,\exists v_1,v_2\in\Sigma^\ast: v_1v_2=v, uv^{k_1}v_1\in \llbracket t_i\rrbracket, v_2v^{k_2+k_1}v_1\in\llbracket s_i \rrbracket \\
    &\iff \exists k_1,k_2\geq 0,\exists v_1,v_2\in\Sigma^\ast: v_1v_2=v, (uv^{k_1}v_1,v_2v^{k_2+k_1}v_1)\in\llbracket \tau \rrbracket_\circ.\qedhere
  \end{align*}
\end{proof}

\begin{cor}\label{test}
  Let $\tau\in\text{Exp}_\circ$. Then
  \begin{enumerate}[label={\arabic*.},ref={6.13.\arabic*}]
    \item $\forall (u,v)\in\llbracket\Gamma(\tau)\rrbracket_\circ, \exists (u',v')\in\llbracket\tau\rrbracket_\circ\colon (u',v')\to^\ast_\gamma (u,v)$, \label{cor:connectionRhoGammaRewrite}
    \item $\{uv^\omega\mid (u,v)\in \llbracket\tau\rrbracket_\circ\} = \{uv^\omega\mid (u,v)\in \llbracket\Gamma(\tau)\rrbracket_\circ\}$ and \label{cor:tauGammaSameOmegaSemantics}
    \item $\llbracket\tau\rrbracket_\circ\subseteq \llbracket\Gamma(\tau)\rrbracket_\circ$. \label{cor:tauSubsetGamma}
  \end{enumerate}
\end{cor}

\begin{proof}
  In Proposition
  \ref{prop:connectionRhoGamma} we have
  $(uv^{k_1}v_1,v_2v^{k_2+k_1}v_1)\to^\ast_\gamma (u,v)$, hence $(1)$
  follows. Moreover, this shows that for each lasso in
  $\llbracket\tau\rrbracket_\circ$ there exists a $\gamma$-equivalent
  one in $\llbracket\Gamma(\tau)\rrbracket_\circ$ and vice-versa,
  hence taking direct images via $f$ on both sides yields the same
  set, giving us $(2)$. Finally, for
  $(u,v)\in\llbracket\tau\rrbracket_\circ$. Choose $k_1=k_2=0$,
  $v_1=\epsilon$ and $v_2=v$, so
  $(u,v)=(uv^{k_1}v_1,v_2v^{k_2+k_1}v_1)\in\llbracket\tau\rrbracket_\circ$.
  By Proposition \ref{prop:connectionRhoGamma},
  $(u,v)=(u,v_1v_2)\in\llbracket\Gamma(\tau)\rrbracket_\circ$
  establishing $(3)$.
\end{proof}

For technical reasons, we require an additional lemma before we can
conclude that for a given rational $\omega$-expression $T$,
$\Gamma(h(T))$ is a representation of $T$. More concretely, given
an arbitrary rational lasso expression $\tau$, it is not the case
that the semantics of $\Gamma(\tau)$ is closed under $\gamma$-reduction.
As a counterexample, consider the rational lasso expression $(aaa,a)$.
Then $\Gamma(aaa,a)=(aa,a)+(aaa,a)$ whose semantics is not closed
under $\gamma$-reduction (it misses both $(\epsilon,a)$ and $(a,a)$).
The reason for this is that the intersection
operation has a limit to how much it can shift underneath the $(-)^\circ$.
In the example with $(aaa,a)$, one could apply $\Gamma$ several times,
and indeed after a few consecutive applications we obtain a rational lasso
expression whose semantics is closed under $\gamma$-reduction. Luckily, if
one starts out with a rational lasso expression whose semantics is already
closed under $\gamma$-expansion, then this is no longer a problem and applying
$\Gamma$ yields a rational lasso expression whose semantics is also closed
under $\gamma$-reduction.

\begin{lem}\label{lem:tauClosedExpansionAlsoGamma}
  Let $\tau\in\text{Exp}_\circ$. If $\llbracket \tau\rrbracket_\circ$ is closed
  under $\gamma$-expansion, then $\llbracket\Gamma(\tau)\rrbracket_\circ$
  is $\sim_\gamma$-saturated.
\end{lem}

\begin{proof}
  By Lemma \ref{lem:saturationLemma}, $\llbracket\Gamma(\tau)\rrbracket_\circ$ 
  is $\sim_\gamma$-saturated if it is closed under both $\gamma$-reduction
  and $\gamma$-expansion so this is what we show.
  
  We begin by showing that $\llbracket\Gamma(\tau)\rrbracket_\circ$ is closed
  under $\gamma$-reduction.
  We split the proof into two cases, one for each type of reduction.
  For $\gamma_1$-reduction we look at $(ua,va)\to_{\gamma_1}(u,av)$.
  Suppose $(ua,va)\in\llbracket\Gamma(\tau)\rrbracket_\circ$, we want to
  show that $(u,av)\in\llbracket\Gamma(\tau)\rrbracket_\circ$.
  By Proposition \ref{prop:connectionRhoGamma}, we find $k_1,k_2,v_1,v_2$ such
  that
  \[
    va=v_1v_2 \qquad \text{ and } \qquad (ua(va)^{k_1}v_1,v_2(va)^{k_2+k_1}v_1)\in \llbracket\tau\rrbracket_\circ.
  \]
  We proceed by case analysis on $v_2$. If $v_2=\epsilon$, then
  $v_1=va$ and
  \[
    (ua(va)^{k_1}va,(va)^{k_2+k_1}va)\in\llbracket\tau\rrbracket_\circ.
  \]
  As $\llbracket\tau\rrbracket_\circ$ is closed under $\gamma$-expansion, we
  also have that
  \[
    (u(av)^{k_1+1}a,v(av)^{(2\cdot k_2+k_1)+(k_1+1)}a) = (ua(va)^{k_1}va,((va)^{k_2+k_1}va)^2)\in\llbracket\tau\rrbracket_\circ.
  \]
  Using Proposition \ref{prop:connectionRhoGamma}, we obtain
  $(u,av)\in\llbracket\Gamma(\tau)\rrbracket_\circ$. If $v_2\not =\epsilon$,
  we find $v_2'$ with $v_2=v_2'a$. Then
  \[
    (u(av)^{k_1}av_1,v_2'(av)^{k_2+k_1}av_1)=(ua(va)^{k_1}v_1,v_2(va)^{k_2+k_1}v_1)\in\llbracket\tau\rrbracket_\circ.
  \] Hence, by Proposition \ref{prop:connectionRhoGamma} and as
  $v_1v_2'=v$ we again have $(u,av)\in\llbracket\Gamma(\tau)\rrbracket_\circ$.

  For $\gamma_2$-reduction we look at $(u,v^k)\to_{\gamma_2}(u,v)$ where
  $k\geq 1$. Assume that $(u,v^k)\in\llbracket\Gamma(\tau)\rrbracket_\circ$.
  So by Proposition \ref{prop:connectionRhoGamma} there are $k_1,k_2,v_1,v_2$
  with
  \[
    v^k=v_1v_2 \qquad \text{ and } \qquad (u(v^k)^{k_1}v_1,v_2(v^k)^{k_2+k_1}v_1)\in \llbracket\tau\rrbracket_\circ.
  \]
  As $v_1v_2=v^k$ we can find some $w_1,w_2$ and $\ell_1,\ell_2$ with
  $v^{\ell_1}w_1=v_1$ and $w_2v^{\ell_2}=v_2$. It follows that
  \[
    (uv^{k\cdot k_1 + \ell_1}w_1,w_2v^{\ell_2+k\cdot (k_2+k_1)+\ell_1}w_1) = (u(v^k)^{k_1}v^{\ell_1}w_1,w_2v^{\ell_2}(v^k)^{k_2+k_1}v^{\ell_1}w_1)\in\llbracket\tau\rrbracket_\circ.
  \] Hence, by Proposition \ref{prop:connectionRhoGamma},
  $(u,v)=(u,w_1w_2)\in\llbracket\Gamma(\tau)\rrbracket_\circ$. So
  $\llbracket\Gamma(\tau)\rrbracket_\circ$ is closed under $\gamma$-reduction.

  Next we show that $\llbracket\Gamma(\tau)\rrbracket_\circ$ is closed
  under $\gamma$-expansion.
  Let $(u,v)\in\llbracket\Gamma(\tau)\rrbracket_\circ$ and
  $(u_1,v_1)\to_\gamma (u,v)$. We want to show that
  $(u_1,v_1)\in\llbracket\Gamma(\tau)\rrbracket_\circ$.
  By Corollary \ref{cor:connectionRhoGammaRewrite} there exists
  some $(u_2,v_2)\in\llbracket\tau\rrbracket_\circ$ with
  $(u_2,v_2)\to^\ast_\gamma (u,v)$. As $(u_1,v_1)\sim_\gamma(u_2,v_2)$, there
  exists a lasso $(u_3,v_3)$ such that
  $(u_3,v_3)\to^\ast_\gamma (u_1,v_1)$ and $(u_3,v_3)\to^\ast_\gamma (u_2,v_2)$.
  As $\llbracket\tau\rrbracket_\circ$ is closed under $\gamma$-expansion,
  $(u_3,v_3)\in\llbracket\tau\rrbracket_\circ$. It follows by
  Corollary \ref{cor:tauSubsetGamma} that
  $(u_3,v_3)\in\llbracket\Gamma(\tau)\rrbracket_\circ$. Finally, as
  we have already shown that $\llbracket\Gamma(\tau)\rrbracket_\circ$ is
  closed under $\gamma$-reduction and as
  $(u_3,v_3)\to^\ast_\gamma (u_1,v_1)$, we obtain that
  $(u_1,v_1)\in\llbracket\Gamma(\tau)\rrbracket_\circ$.
\end{proof}

\begin{prop}\label{prop:rationalLassoFromOmega}
  Let $T\in \text{Exp}_\omega$. Then $\Gamma(h(T))$ represents
  $T$.
\end{prop}

\begin{proof}
  By Proposition \ref{prop:weakCoverMapYieldsWeakCovering}, $h(T)$ weakly represents $T$, i.e.
  $\text{UP}(\llbracket T\rrbracket_\omega) = \{uv^\omega\mid (u,v)\in\llbracket h(T)\rrbracket_\circ\} $.
  It follows from Corollary \ref{cor:tauGammaSameOmegaSemantics} that $\Gamma(h(T))$ also weakly
  represents $T$, as
  \[
    \text{UP}(\llbracket T\rrbracket_\omega) = \{uv^\omega\mid (u,v)\in\llbracket h(T)\rrbracket_\circ\} = \{uv^\omega\mid (u,v)\in\llbracket\Gamma(h(T))\rrbracket_\circ\} .
  \] Furthermore, as $\llbracket h(T)\rrbracket_\circ$ is closed under
  $\gamma$-expansion (Proposition \ref{prop:weakCoverMapExpansionClosed}),
  $\llbracket\Gamma(h(T))\rrbracket_\circ$
  is $\sim_\gamma$-saturated (Proposition
  \ref{lem:tauClosedExpansionAlsoGamma}).
  Hence, by Remark \ref{rmk:rmk1}
  $\Gamma(h(T))$ represents $T$.
\end{proof}

The previous proposition together with the Brzozowski
construction for lasso automata give us the main result of this
section. This establishes the remaining arrow in Figure
\ref{fig:contributions} from the introduction.

\begin{thm}\label{thm:omegaAutomatonFromOmegaExpression}
  Every rational $\omega$-language is accepted by a finite $\Omega$-automaton.
\end{thm}

\begin{proof}
  Let $L$ be a rational $\omega$-language and $T\in \text{Exp}_\omega$
  a rational $\omega$-expression denoting it.  Then the lasso automaton
  $(\widehat{\mathcal{C}},[\Gamma(h(T))]_{\sim_C})$ is a finite
  $\Omega$-automaton with
  \[
    L_\omega (\widehat{\mathcal{C}},[\Gamma(h(T))]_{\sim_C}) = \llbracket T\rrbracket_\omega = L.\qedhere
  \]
\end{proof}

\section{Conclusion}
We have introduced rational lasso expressions and languages and shown
a Kleene Theorem for lasso languages and $\omega$-languages. In order
to obtain these results we gave a Brzozowski construction for lasso
automata. Moreover, we introduced the notion of representation and
showed how to construct a representing rational lasso expression from
a rational $\omega$-expression. As a consequence, we obtained a
new construction for converting rational $\omega$-expressions to
$\Omega$-automata.

Our results motivate some interesting directions for future
work. In
\cite{angluin:2016:learningRegularOmegaLanguages}, Angluin \etal\
introduce syntactic and recurring FDFAs and
show that they can be up to exponentially smaller than periodic
FDFAs. This raises the question whether from a given rational $\omega$-expression,
we can construct a rational lasso expression such that the Brzozowski
construction yields a syntactic or recurring FDFA. These FDFAs operate using
different congruences and it might be of interest to see how the structure
of rational lasso expressions interact with them. As one of the applications
is language learning, it is also worth investigating the complexity of our
constructions. The state complexity of the root and intersection operation
are known. We remark that the languages of which we compute the root
are always transitive and we wonder if the state complexity of the root
operation improves when restricting to transitive languages. Moreover, as
Calbrix \etal\ devised $L_{\$}$ in the hopes of improving algorithms for
model checking and verification, it would be of interest to establish a
more direct link between lasso automata and logics such as linear temporal
logic. This naturally leads on to a comparison between the more traditionally
used $\omega$-automata, such as nondeterministic B\"{u}chi automata, and
$\Omega$-automata.

Another line of work is to show completeness of the lasso algebra we introduced
in Section \ref{sec:3}, and to look more closely at the links between lasso
and Wagner algebras.

\paragraph{Acknowledgements.} The author would like to thank Yde Venema for
suggesting this topic and also both Tobias Kapp\'{e} and Yde Venema for
valuable discussions. Furthermore, the author
would like to thank Harsh Beohar and Georg Struth for valuable discussions
and for reading earlier drafts of this paper.

\bibliographystyle{plain} %also have named
\bibliography{references}

\end{document}